\documentclass{IEEEtran4PSCC}

\usepackage{cite}
\usepackage[pdftex]{graphicx}
\usepackage{epsfig} 
\usepackage{amsmath} 
\usepackage{amssymb}  
\usepackage{amsfonts}
\usepackage{amsthm}
\usepackage{mathtools}  
\usepackage{braket}
\usepackage{optidef,booktabs}
\usepackage{cite}
\usepackage{url}
\usepackage{lipsum} 
\usepackage{algorithm}
\usepackage{algpseudocode}
\usepackage[utf8]{inputenc} 
\usepackage[font=small,skip=1pt]{subcaption}
\usepackage[font=small,skip=1pt]{caption}
\usepackage{tabularray}
\usepackage{tikz}
\usetikzlibrary{positioning,calc,arrows.meta}
\usepackage{circuitikz}
\usepackage{xcolor}
\usepackage{nicefrac}
\usepackage{hyperref}
\usepackage{bm}
\usetikzlibrary{shapes.symbols}

\hypersetup{ 		
	colorlinks=true,        
	linkcolor=blue,
	urlcolor=blue, 
	citecolor=blue     		
}
\def \ebr#1{{\color{blue}#1}} 
\def \ebr#1{{\color{red}#1}} 

\providecommand{\abs}[1]{\lvert#1\rvert}
\providecommand{\norm}[1]{\lVert#1\rVert}
\providecommand{\SRG}[1]{\operatorname{SRG}(#1)}

\newtheorem{remark}{Remark}
\newtheorem{example}{Example}

\newtheorem{lemma}{Lemma}
\newtheorem{theorem}{Theorem}

\newtheorem{definition}{Definition}
\newtheorem{fact}{Property}

\usepackage{tikz}
\usepackage{circuitikz}
\usepackage{graphics,graphicx} 

\usetikzlibrary{arrows.meta,    
                calc, chains,   
                positioning,    
                quotes          
                }
 \definecolor{color1bg}{HTML}{77AC30}
\definecolor{color2bg}{HTML}{D95319}
\definecolor{color3bg}{HTML}{0072BD}
\definecolor{color4bg}{HTML}{A2142F}
\definecolor{backgroundblue}{rgb}{0.6, 0.8, 1}
\definecolor{backgroundorange}{rgb}{1, 0.752, 0}

\tikzset{sin v source/.style={circle,draw,append after command={
    \pgfextra{
    \draw
      ($(\tikzlastnode.center)!0.5!(\tikzlastnode.west)$)
       arc[start angle=180,end angle=0,radius=0.425ex] 
      (\tikzlastnode.center)
       arc[start angle=180,end angle=360,radius=0.425ex]
      ($(\tikzlastnode.center)!0.5!(\tikzlastnode.east)$) ;}},scale=1.5,}}
      
\usepackage{multicol}

\tikzset{
  v_sin/.pic={
  \draw (0,0) sin (1,1) cos (2,0) sin (3,-1) cos (4,0);
  \draw (2,0) circle (2.5);
  }
}
\tikzset{
  abc_dq/.pic={
  \draw (-0.5,0) rectangle (0.5,-0.7);
  \draw (-0.5,0) -- (0.5,-0.7);
  \node at (-0.2,-0.5) { $dq$};
  \node at (0.2,-0.2) {$abc$};
  }
}

\makeatletter
\let\old@ps@headings\ps@headings
\let\old@ps@IEEEtitlepagestyle\ps@IEEEtitlepagestyle
\def\psccfooter#1{%
    \def\ps@headings{%
        \old@ps@headings%
        \def\@oddfoot{\strut\hfill#1\hfill\strut}%
        \def\@evenfoot{\strut\hfill#1\hfill\strut}%
    }%
    \def\ps@IEEEtitlepagestyle{%
        \old@ps@IEEEtitlepagestyle%
        \def\@oddfoot{\strut\hfill#1\hfill\strut}%
        \def\@evenfoot{\strut\hfill#1\hfill\strut}%
    }%
    \ps@headings%
}
\makeatother

\psccfooter{%
        \parbox{\textwidth}{\hrulefill \\ \small{24th Power Systems Computation Conference} \hfill \begin{minipage}{0.2\textwidth}\centering \vspace*{4pt} \includegraphics[scale=0.06]{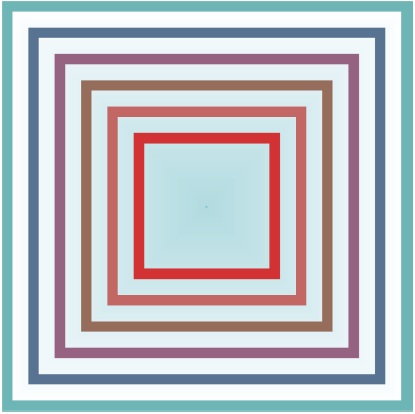}\\\small{PSCC 2026} \end{minipage} \hfill \small{Limassol, Cyprus --- June 8 -- June 12, 2026}}%
}

\begin{document}
%


\title{On the Impact of Operating Points on Small-Signal Stability: Decentralized Stability Sets via Scaled Relative Graphs}

\author{
\IEEEauthorblockN{Eder Baron-Prada$^{1,2}$, Adolfo Anta$^{1}$, Florian Dörfler$^{2}$}
\IEEEauthorblockA{$^1$Austrian Institute of Technology (AIT), 1210 Vienna, Austria }
\IEEEauthorblockA{$^2$Automatic Control Laboratory, ETH Zurich, 8092 Zurich, Switzerland}
}

\maketitle

\begin{abstract}
This paper presents a decentralized frequency-domain framework to characterize the influence of the operating point on the small-signal stability of converter-dominated power systems. The approach builds on Scaled Relative Graph (SRG) analysis, extended here to address Linear Parameter-Varying (LPV) systems. By exploiting the affine dependence of converter admittances on their steady-state operating points, the centralized small-signal stability assessment of the grid is decomposed into decentralized, frequency-wise geometric tests. Each converter can independently evaluate its feasible stability region, expressed as a set of linear inequalities in its parameter space. The framework provides closed-form geometric characterizations applicable to both grid-following (GFL) and grid-forming (GFM) converters, and validation results confirm its effectiveness.
\end{abstract}


\begin{IEEEkeywords}
Scaled Relative Graphs, Decentralized Stability.
\end{IEEEkeywords}


\section{Introduction}
 Renewable energy resources are now central elements of modern grids~\cite{Milano2018}, displacing synchronous generators with converters operating at millisecond time scales~\cite{Milano2018}. This shift introduces new small-signal stability challenges: converters interact strongly with network impedances and with each other, often giving rise to oscillations and resonance phenomena~\cite{dong2023analysis,Cheng2023}. 
Moreover, even slight variations in converter operation setpoints, grid impedance, or controller gains can significantly shift small-signal stability margins.

In practice, the integration of renewable plants still relies on interconnection studies, where stability is assessed by simulating different converter setpoints and grid conditions~\cite{nerc}. Although valuable, this approach is both computationally demanding and incomplete. The difficulty is compounded by the \emph{curse of dimensionality}: as the number of converters increases, the operating setpoint space becomes high-dimensional, and the number of simulations required to cover it grows exponentially. For instance, ten converters with six setpoints each already produce sixty million possible combinations. Even with substantial computational resources, such studies can only sample this space sparsely, leaving small-signal stability boundaries largely unexplored~\cite{birchfield2023review}. This challenge is pronounced in renewable-rich grids, where operating points shift rapidly in response to changes in generation and demand.

Traditional small-signal stability tools, such as eigenvalue analysis, generalized Nyquist criterion and time-domain simulations, remain useful but exhibit fundamental limitations in this new setting~\cite{Tao2022,Fan2020_problemsAdmittance}.  These centralized approaches typically yield only an incomplete outcome: stable or unstable for some chosen operation parameters.  Moreover, as the number of converters increases, the dimensionality of system models grows rapidly, making centralized analysis computationally intensive and often impractical\cite{Molinas2020}. 

To mitigate these two issues, several \emph{decentralized stability certificates} have been proposed to assess converter-grid interactions locally. Passivity-based conditions~\cite{Wang2024Limitations,Wang2023_Passivity,Haberle2025}, small-gain and phase-margin criteria~\cite{huang2024gain}, and geometric-based approaches\cite{Baron2025decentralized,Feng2025} offer scalable and interpretable stability guarantees. These certificates enable decentralized verification of small-signal stability by imposing local frequency-domain conditions on each converter. However, most existing methods apply only to linear time-invariant (LTI) models and fixed operating points, and thus cannot capture the parameter dependence intrinsic to converter behavior. {Hence, it provides only limited insight into robustness with respect to different operating points.} 

What is needed are analytical tools that can characterize stability margins across operating ranges. This work addresses this gap by developing a systematic framework tailored to converter-dominated grids. Because converter dynamics vary with operating setpoints, they are naturally described by LPV models. Building on this observation, we extend frequency-wise SRGs, a novel tool for analyzing LTI systems~\cite{Baron2025SRG,Chen2025,Baron2025SGP,zhang2025phantom,baronprada2026powergrids}, to LPV systems. This extension enables the direct characterization of small-signal stability across operating ranges without resorting to exhaustive scenario enumeration.

Furthermore, we introduce geometric descriptions of decentralized stability sets, which capture feasible operating regions as functions of system parameters. These sets define secure-operation boundaries and characterize small-signal stability margins. 
They also support corrective actions when operation approaches instability, as operators can adjust converter setpoints to remain within certified regions. The proposed framework is validated for both GFM and GFL converters, demonstrating its applicability across control architectures.

\section{Preliminaries}\label{subsec:basics}
\subsection{Basic Notation}
The mathematical foundation of this work builds upon LTI and LPV system representations. We consider LTI systems with state-space representation:
\begin{align*}
\dot{x} = Ax + B\hat{u},\quad \hat{y} = Cx + D\hat{u},
\end{align*}
where $x \in \mathbb{R}^m$ represents the state vector, and the matrices $(A,B,C,D)$ have appropriate dimensions. The corresponding transfer function $H(s) = C(sI-A)^{-1}B+D$ characterizes the input-output mapping under zero initial conditions \cite{zhou1998}. We focus solely on square systems (systems with an equal number of inputs and outputs) as is the case for power systems modeled via impedances. For parameter-dependent systems, we resort to LPV systems, where we consider affine LPV transfer function matrix of the form
\begin{align}
    H(\gamma,s)=H_{0}(s)+{ \sum_{\forall k}}\gamma_{k}H_{k}(s).
    \label{eqn:affine_LPV} 
\end{align} 
where the  parameter vector $\gamma$ encodes system parameters. We focus on $\mathcal{RH}_\infty$, the space of rational, proper, and stable transfer functions, which can represent admittance matrices \cite{Fan2020_problemsAdmittance}. Let $F \in \{\mathbb{R}, \mathbb{C}\}$ denote the field of interest and define $\mathcal{L}_2^m(F)$ as the space of square-integrable signals $\hat{u}, \hat{y}: \mathbb{R}_{\ge 0} \rightarrow F^m$ equipped with the inner product $\langle \hat{u}, \hat{y} \rangle := \int_0^\infty \hat{u}(t)^* \hat{y}(t) \, dt$ and induced norm $\|\hat{u}\|_2 := \sqrt{\langle \hat{u}, \hat{u} \rangle}$.  The imaginary unit is denoted by $\textup{j}$.   {The Fourier transform of $\hat{u} \in \mathcal{L}_2^n({F})$ is defined as $u(\textup{j}\omega) := \int_0^{\infty} e^{-\textup{j}\omega t} \hat{u}(t) \, dt$.}
The distance between two sets $A, B \subset \mathbb{C}$ is defined as $\operatorname{d}(A,B)=\inf_{a\in A,b\in B}\norm{a-b}_{\mathbb{C}}$, where $\norm{a-b}_{\mathbb{C}}=\sqrt{(\operatorname{Re}(a-b))^2+(\operatorname{Im}(a-b))^2}$. The Minkowski sum and difference are $A\oplus B=\{a+b\mid a\in A,b\in B\}$ and $A\ominus B=\{a-b\mid a\in A,b\in B\}$, and the convex hull $\operatorname{co}(A)$ is the smallest convex set containing $A$.  {The angle between two nonzero vectors $u,y$ in the same Hilbert space is defined as $\angle(u,y):= \arccos\!\left( \tfrac{\Re\langle u,\, y \rangle} {\|u\|_2\,\|y\|_2} \right)$}.

\subsection{Scaled Relative Graph Theory}

Initially developed for convergence analysis in optimization \cite{ryu2022large}, the SRG provides a geometric framework for characterizing operator properties. It has been extended to the study of linear and nonlinear dynamical systems \cite{Chaffey_2023}, offering a set-based representation of frequency-domain interactions that enables a precise stability assessment. For a transfer function matrix $H(s)$, the frequency-wise SRG is defined as
\begin{align}
\operatorname{SRG}(H(\textup{j}\omega)) = \left\{ \dfrac{\|H(\textup{j}\omega)u\|_2}{\|u\|_2} e^{\pm \mathrm{j}  \angle(u,H(\textup{j}\omega)u)} \right\}\subset\mathbb{C},
\label{eqn:SRG_operator_transferfunction}
\end{align}
for each $\omega \in [0,\infty)$, evaluated over all $u$ such that $\|u\|_2 \neq 0$ and $\|H(\textup{j}\omega) u\|_2 \neq 0$. The complete $\operatorname{SRG}(H(\textup{j}\omega))$ is obtained as the union of all frequency-wise SRGs, forming a three-dimensional set. The SRG framework has several properties that facilitate its application to stability analysis. 

\begin{fact}[Chord Property\cite{Chaffey_2023,krebbekx2025}]\label{property:chord}
An operator $A$ is said to satisfy the chord property if for every bounded $ z \in \operatorname{SRG}(A)$, the line segment $[z, z^*]$, defined as $z_1,z_2 \in \mathbb{C}$ as $[z_1,z_2] := \{ \beta z_1 + (1 - \beta)z_2 \mid \beta \in [0,1] \}$, is in $\operatorname{SRG}(A)$. 
\end{fact}

This property facilitates the analysis of operator combinations, as established by the following result.

\begin{fact}[Sum of operators\cite{Chaffey_2023,krebbekx2025}]\label{property:sum}
Let $A$ and $B$ be bounded operators. If either $\SRG A$ or $\SRG B$ meet the chord property. Then, $\SRG{A+B}\subseteq\SRG A + \SRG B$.
\end{fact}

We introduce approximation schemes that meet the \emph{chord property} and simplify computational needs. We define the \textit{tight chord approximation} that extends the SRG to its convex hull.


\begin{definition}[Tight chord approximation]
\label{def:approx_chord}
Given $H(\mathrm{j}\omega)$, its tight chord approximation $\overline{H}(\mathrm{j}\omega)$ is defined at each $\omega$ as
\begin{align*}
    \operatorname{SRG}(\overline{H}(\mathrm{j}\omega))= \operatorname{co}\!\big(\operatorname{SRG}(H(\mathrm{j}\omega))\big),
\end{align*}
i.e., $\operatorname{SRG}(\overline{H}(s))$ is the frequency-wise convex hull of $\operatorname{SRG}(H(s))$. 
\end{definition}

Moreover, we also introduce the disk approximation, which provides a conservative outer approximation with guaranteed satisfaction of the chord property.

\begin{definition}[Disk approximation]
\label{def:approx_disk}
    Given $H(\textup{j}\omega)$, $\widehat{H}(\textup{j}\omega)$ is defined as an operator such that  $ \operatorname{SRG}(\widehat{H}(\textup{j}\omega))\supseteq\operatorname{SRG}({{H}(\textup{j}\omega)})$ and there are $a\in\mathbb{R}$ and $b>0$, so that
\begin{align*}
     \operatorname{SRG}(\widehat{H}(\textup{j}\omega)):=
     \{r_de^{\textup{j}\alpha}+a |\;\forall \;\alpha \in [0,2\pi], \;b \geq r_d > 0\}. 
\end{align*}
\end{definition}

In essence, $\widehat{H}(\textup{j}\omega)$ over-approximates ${H}(\textup{j}\omega)$ as a closed disk satisfying the chord property, whose parameters $a$ and $b$ determine its tightness. Note that the \emph{tight chord} approximation is unique, whereas the disk approximation is not unique and $\operatorname{SRG}(\widehat{H}(\textup{j}\omega))\supseteq\operatorname{SRG}(\overline{H}(\textup{j}\omega))$.

\begin{example}
Consider the transfer function representing an RL admittance:    \begin{align}\label{eqn:Ygrid}
        Y_{\text{grid}}(s)=\begin{bmatrix}
        1.281+\tfrac{3.14 s}{\omega_0}     & -3.14 \\
        3.14 & 1.28+\tfrac{3.14s}{\omega_0}
        \end{bmatrix}^{-1}, 
    \end{align} 
    where $\omega_0=2\pi50$.  {We compute the SRG using \eqref{eqn:SRG_operator_transferfunction} for each frequency, where we use $H(s)=Y_{\text{grid}}(s)$}. The SRG and its frequency-wise \emph{tight chord} approximation are shown in Fig.~\ref{fig:example_grid}.
\begin{figure}[ht]
    \centering
\begin{subfigure}[b]{0.24\textwidth}
    \centering
    \includegraphics[width=0.8\linewidth]{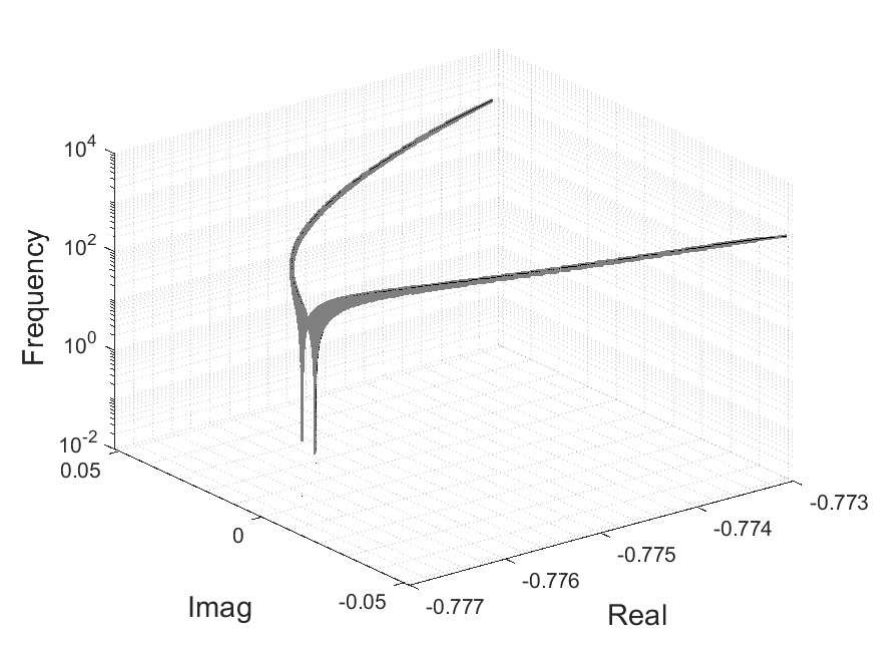}
    \caption{}
\end{subfigure}
\begin{subfigure}[b]{0.24\textwidth}
    \centering
    \includegraphics[width=0.8\linewidth]{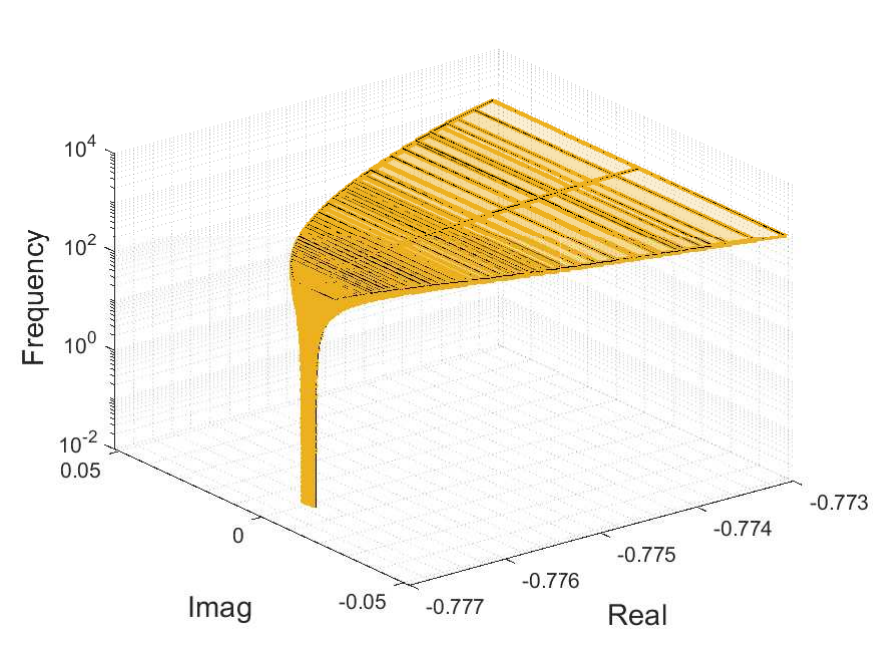}
    \caption{}
\end{subfigure}
    \caption{(a)~$\operatorname{SRG}(Y_{\text{grid}}(\textup{j}\omega))$ for all $\omega\in[10^{-2},10^3]$ (b)~$\operatorname{SRG}(\overline{Y}_{\text{grid}}(\textup{j}\omega))$ for all $\omega\in[10^{-2},10^3]$.}
    \label{fig:example_grid}
\end{figure}
\end{example}

\section{Power System Model} \label{sec:model}
\subsection{Multi-Converter System Representation}

The analytical framework considers $\mathcal{C}=\{1,\ldots,n\}$ converters  interconnected through a grid, as depicted in Fig.~\ref{fig:multi_conv_system}. Different control architectures across converters are used in multi-converter grids, including GFL {, which uses phase-locked loops (PLL) for grid synchronization, and GFM  that actively sets voltage and frequency references\cite{Milano2018}}.

\begin{figure}[hbt]
\vspace{-4mm}
    \centering 
    \begin{tikzpicture}[scale=1, every node/.style={transform shape}]
\draw[ thick,color=color3bg] (6.6,1.5) node[left]{IBR$_{\text{GFL}}$};
\draw[ thick,color=color2bg] (6.6,2.5) node[left]{IBR$_{\text{GFM}}$};
\draw[ thick,color=color2bg] (-0.4,1.5) node[left]{IBR$_{\text{GFM}}$};
\draw[ thick,color=color3bg] (-0.9,2.5) node[left]{IBR$_{\text{GFL}}$};


\draw
(4.5,1) coordinate(1-node) 
(1-node) --++ (0,0.3) node[above]{} 
(1-node)--++(0,0.8) node[left]{$n$} 
coordinate[pos=0.3](1-r) 
coordinate[pos=0.6](1-rr)
coordinate[pos=0.9](1-rrr)
;
\draw (5,1.5)node[sacdcshape,scale=-.5,color=color3bg,fill=white] (DC1){};

\draw
(0.5,1)  coordinate(2-node) 
(2-node) --++ (0,0.3) node[above]{}
(2-node)--++(0,0.8) node[right]{$2$} 
coordinate[pos=0.3](2-r) 
coordinate[pos=0.6](2-rr)
coordinate[pos=0.9](2-rrr)
;
\draw (-0.65,2.5)node[sacdcshape,scale=-.5,color=color3bg,fill=white] (DC1){};

\draw
(0,2)  coordinate(3-node) 
(3-node) --++ (0,0.3) node[above]{}
(3-node)--++(0,0.8) node[left]{$1$} 
coordinate[pos=0.3](3-r) 
coordinate[pos=0.6](3-rr)
coordinate[pos=0.9](3-rrr)
;
\draw (0,1.5)node[sacdcshape,scale=-.5,color=color2bg,fill=white] (DC2){};
\draw
(4.5,2)  coordinate(5-node) 
(5-node)--++(0,0.3) node[above]{} 
(5-node)--++(0,0.8) node[left]{$n$-$1$} 
coordinate[pos=0.3](5-r) 
coordinate[pos=0.6](5-rr)
coordinate[pos=0.9](5-rrr);

\draw (5,2.5)node[sacdcshape,scale=-.5,color=color2bg,fill=white] (DC2){};


\draw
(1-rr)--++(0.23,0)
(2-rr)--++(-0.23,0) 
(3-rr)--++(-0.43,0)
(5-rr)--++(0.23,0)
(1-rr)--++(-0.8,0)
(2-rr)--++(0.8,0) 
(3-rr)--++(1.5,0)
(5-rr)--++(-0.8,0)
;

\node[cloud, draw,
    fill = gray!10,
    minimum width = 3cm,
    minimum height = 1cm,
    aspect=1.3,
    cloud puffs = 16] (c) at (2.5,2) {Power grid};

\draw[-latex,color3bg] (0.2,2.55) -- (0.65,2.55);
\node at (0.4,2.75) {$ i_{dq,1}$};
\draw[-latex,color3bg] (0.9,2.47) -- (0.9,2.85);
\node at (1,2.9) {$v_{dq,1}$};
\end{tikzpicture}
    \caption{Configuration of the multi-converter grid.}
    \label{fig:multi_conv_system}
    \vspace{-4mm}
\end{figure}
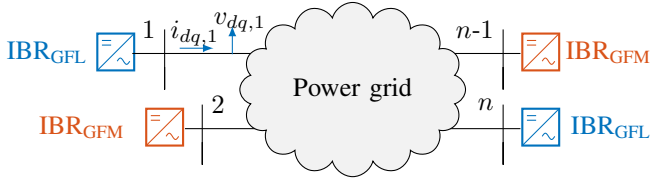
Independent of the control implementation, the linearized dynamic behavior of each converter unit can be systematically characterized through its frequency-domain admittance representation~\cite{huang2024gain,Fan2020_problemsAdmittance}. The admittance representation is a relationship between the vectors $[ i_{d,i} \; i_{q,i} ]^\top$ and $[  v_{d,i} \;  v_{q,i} ]^\top$ that describe the behavior of the current and voltage of the $i$-th converter unit, respectively, expressed within a global $dq$-coordinate. The converter admittance can be formulated as an LPV system by explicitly expressing its dependence on the parameters. These parameters can be chosen to represent operating points, control gains, or filter parameters. In this work we focus on the impact of the operating points on stability, although the framework is amenable to any other parameter.  Since the converter dynamics are nonlinear (see Appendix~\ref{Appx:GFL_admittance}), we linearize around the operating point. For the $i$-th converter, the admittance transfer function is:
\begin{align}\label{eqn:admittance_model}
   \begin{bmatrix} i_{d,i} \\ i_{q,i} \end{bmatrix} =  {-J(\theta_i)} {Y}_{i}(\gamma_i,s) {J(-\theta_i)}\begin{bmatrix}  v_{d,i} \\  v_{q,i} \end{bmatrix} \,,
\end{align}
where ${Y}_{i}(\gamma_i,s)$ is the $2 \times 2$ transfer matrix of the $i$-th converter, with $\gamma_i=[\gamma_{i,1},\dots,\gamma_{i,k},\dots,\gamma_{i,n_p}] \in \phi_i$ denoting its operating parameters over the admissible set $\phi_i$. The parameters $\gamma_i$ are understood to be fixed or they change slowly relative to the dynamics of interest. 
The parameter $\theta_i$ is the steady-state angular displacement between the global $dq$-frame and the converter’s local coordinates, with the  transformation matrix defined as
$${J(\theta_i)} = \begin{bmatrix} \cos \theta_i & -\sin \theta_i \\ \sin \theta_i & \cos \theta_i \end{bmatrix}.$$ 

The derivation methodologies for ${ Y}_{i}(\gamma_i,s)$ under various control strategies, including both GFL and GFM, are extensively documented~\cite{Huang2024_Howmany,Wang2024Limitations} and summarized in Appendix~\ref{Appx:GFL_admittance}. These admittance matrices are formulated in local coordinates and then referenced to global $dq$-coordinates. The individual converter representation extends to the complete $n$-converter system through the aggregate formulation:
\begin{align}\label{eqn:admittance_devices}
    \underbrace{ \begin{bmatrix} i_{d,1} \\ i_{q,1}  \\ \vdots \\ i_{d,n} \\ i_{q,n} \end{bmatrix}}_{=: {\bf i}} = 
  \; \underbrace{-J (\Theta) \; {\bf Y}({\Gamma},s) \; J(-\Theta)}_{=:\Tilde{\bf Y}(\Gamma,s)} \;
\underbrace{\begin{bmatrix}  v_{d,1} \\  v_{q,1} \\  \vdots \\  v_{d,n} \\  v_{q,n} \end{bmatrix}}_{=:{\bf v}} \,,
\end{align}
\noindent
 The admittance matrix ${\bf Y}(\Gamma,s)$ represents a $2n \times 2n$ block-diagonal structure constructed from individual converter admittances: ${\bf Y}(\Gamma,s) = {\rm diag}\{{Y_{1}(\gamma_1,s),\dots,Y_{n}(\gamma_n,s)}\}.$ 
Moreover, $\Gamma = [\gamma_1, \ldots, \gamma_n]$, and $J({\Theta}) = {\rm diag}\{{J(\theta_1)},\dots,{J(\theta_n)}\} $ provides the coordinate transformation matrix incorporating all angular displacements.
\subsection{Grid Dynamics }\label{sec:grid}
The electrical behavior of each line between terminals $i$ and $j$ is modeled in the frequency domain using per-unit normalization, with the current–voltage relation expressed in matrix form within the synchronously rotating frame at nominal frequency $\omega_0$:
\begin{align*}
    \begin{bmatrix}
          i_{d,ij}(s)\\  i_{q,ij}(s)
    \end{bmatrix} = y_{ij}(s)
    \left(
    \begin{bmatrix}
          v_{d,i}(s)\\   v_{q,i}(s)
    \end{bmatrix} - \begin{bmatrix}
          v_{d,j}(s)\\   v_{q,j}(s)
    \end{bmatrix}
    \right)
\end{align*}
Here, the vector $  i_{dq,ij} = [  i_{d,ij} \;\;   i_{q,ij}]^\top$ quantifies current flowing between terminals $i$ and $j$. The voltage at each terminal is similarly captured by $  v_{dq,i} = [  v_{d,i} \;\;   v_{q,i}]^\top$, while the line characteristics are encoded within the $2 \times 2$ admittance transfer matrix $y_{ij}(s)$. 
%
The grid admittance matrix satisfy $y_{ij}(s) = -y_{ij}(s)$ for distinct terminals $i \neq j$, while self-admittance terms aggregate all connected paths $y_{ii}(s) = \sum_{j \neq i}^n y_{ij}(s)$. 
After obtaining the grid admittance matrix, we perform a Kron reduction \cite{Dorfler2011}, where the maintained nodes are the $n$ nodes, where the converters are connected. This nodal analysis provides a $2n \times 2n$ grid admittance matrix as follows
\begin{align}\label{eqn:Ygrid_closedloop}
 {\bf i} = Y_{\text{grid}}(s){\bf v}.
\end{align}

\subsection{Feedback Loop between Grid and Multi-Converter System}
We can now establish the feedback loop between the converters and the grid. Equations \eqref{eqn:admittance_devices} and \eqref{eqn:Ygrid_closedloop} together capture the power system's closed-loop dynamics, depicted in Fig.~\ref{fig:decentralizedfb}.

\begin{figure}[ht]
\vspace{-4mm}
\centering
\begin{tikzpicture}[scale=1, every node/.style={transform shape}]
\filldraw[color=red!60, fill=red!5, very thick,rounded corners=5] (6,4.6) rectangle (13.25,3.5);
\filldraw[color=black, fill=white] (9.5,4.5) rectangle (11.5,3.7);
\node at (10.5,4.1) {$-{\bf Y}(\Gamma,s)$};
\filldraw[color=black, fill=white] (12,4.5) rectangle (13,3.7);
\node at (12.5,4.1) {$J(\Theta)$};
\filldraw[color=black, fill=white] (8,4.5) rectangle (9,3.7);
\node at (8.5,4.1) {$J(\text{-}\Theta)$};
\filldraw[color=black, fill=white] (9.75,3.4) rectangle (11.25,2.6);
\node at (10.5,3) {$Y^{-1}_{\text{grid}}(s)$}; 
\draw[-latex, line width = .5 pt] (9,4.1) -- (9.5,4.1);
\draw[-latex, line width = .5 pt] (11.5,4.1) -- (12,4.1);
\draw[-latex, line width = .5 pt] (12.5,3.7) -- (12.5,3) -- (11.25, 3 );
\draw[-latex, line width = .5 pt] (9.75,3) -- (8.5,3) -- (8.5,3.7) ;
\node at (9.35,3.1664) {$\bf v$};
\node at (7,3.8) {$\Tilde{\bf Y}(\Gamma,s)$};
\node at (11.6,3.1664) {$\bf i$};
\end{tikzpicture}
\caption{Closed-loop dynamics of a converter-grid system}
\label{fig:decentralizedfb}
\end{figure}
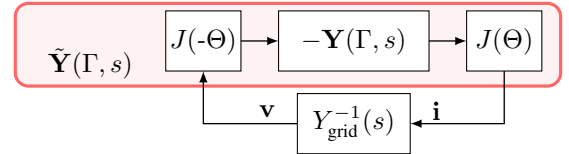 

For a given power demand, multiple combinations of converter setpoints can supply the required power. However, not all points ensure stable closed-loop operation: certain setpoints can trigger resonances or control instabilities, as documented in recent low-inertia grid events~\cite{a2024_diagnosis,dong2023analysis}. Determining the stability boundary for $n$ converters is fundamentally challenging: the poor scalability of centralized analysis makes it computationally prohibitive, while existing eigenvalue- and scenario-based tools offer only limited insight. To overcome this, we employ SRG methods to extract, for each converter, a stable subset $\mathcal{S}_i \subseteq \phi_i$, thereby enabling the systematic certification of stability regions for secure operation. 
We apply SRG-based decentralized stability conditions on the closed-loop system in Fig.~\ref{fig:decentralizedfb}, as formalized in the following theorem.

\begin{theorem}\label{thm:dec_SRG_S} \cite{Baron2025decentralized}
Let ${Y}_{\text{grid}}(s)\in \mathcal{RH}_\infty^{m \times m}$ and  ${\bf \Tilde{Y}}(\Gamma,s)\in \mathcal{RH}_\infty^{m \times m} $ be connected in closed loop as in Fig. \ref{fig:decentralizedfb}, { and both have no poles on $\textup{j}\mathbb{R}\cup \{\infty\}$}.
 If $\forall\; s=\textup{j}\omega \text{, with } \omega\;\in[0,\infty) $, $\forall \; \tau \in (0,1],$ and  $ \forall i\in\mathcal{C}$
\begin{align}
          \operatorname{SRG}(Y_{i}(\gamma_i,s)) \cap -\tau\operatorname{SRG}(\widehat{Y}_{\text{grid}}(s)) = \emptyset,  
          \label{eqn:sufficient_decentralized}
\end{align} 
then  the closed-loop system is $\mathcal{L}_{2}$-stable.
\end{theorem}

Intuitively, the condition states that as long as each converter SRG does not intersect with the disk approximation of $-\tau \SRG{Y_{\text{grid}}(s)}$ at any frequency, the system remains $\mathcal{L}_2$-stable. A related viewpoint was independently introduced in \cite[Proposition 3]{zhang2025phantom}, where the aggregate admittance ${\bf Y}(\Gamma,s)$ is treated as a single system. 
In contrast, our approach decouples the analysis by working directly with individual converter SRGs, enabling decentralized stability characterization. 

\section{LPV SRG and Decentralized Set Estimation }
In this section, we focus on the single-converter case, noting that the results naturally extend to any number of converters connected to the grid. We therefore drop the subindex $i$ and denote the admittance of the $i$-th converter as $Y(\gamma,s)$.

The dynamics of power converters depend on operating points, which motivates extending the SRG framework to LPV systems. Now, consider a converter admittance ${Y}(\gamma,s)$, the \emph{parameter-dependent} SRG is defined as
\begin{align*}
\operatorname{SRG}&(Y(\gamma,s))=\\
&\bigcup_{\gamma\in\phi}
\left\{
  \dfrac{\|Y(\gamma,s) v_{dq}\|_2}{\|v_{dq}\|_2}
 e^{\pm \mathrm{j} \angle(v_{dq},\,Y(\gamma,s) v_{dq})}\right\}\subset\mathbb{C},
\end{align*}
with $s = \mathrm{j}\omega$ where $\omega \in [0,\infty)$. The parameter vector $\gamma$ depends on $i_{dq0} = [i_{d0}, i_{q0}]$ and, when required by the admittance, on higher-order cross terms derived from it (e.g.$i_{d0}^2 i_{q0}^3$). For GFL, the admittance is typically affine in $i_{dq0}$, so $\gamma = [i_{d0}, i_{q0}]$~\cite{Huang2024_Howmany}. In contrast, GFM generally cannot be written in affine LPV form (see Appendix~\ref{Appx:GFL_admittance} and \cite{Huang2024_Howmany}). A practical approach is to approximate the non-affine dependence using series, such as shifted geometric expansions \cite{toth2010modeling}.  {An exception occurs with certain voltage-synchronizing GFM schemes, which are directly LPV-affine without requiring this additional step \cite{andres2023_ipll}}.

\subsection{Decentralized Stability Characterization using SRG}
Given that $Y(\gamma,s)$ is LPV affine in $\gamma$, its SRG can be expressed in a tractable form, as shown in Lemma~\ref{lemma:SRG_LPV_affine}. 

\begin{lemma}\label{lemma:SRG_LPV_affine}
Consider $Y(\gamma,s)$ represented in the affine LPV form of~\eqref{eqn:affine_LPV}, where each $\operatorname{SRG}(Y_{k}(s))$, for $k\neq0$, is approximated by $\operatorname{SRG}(\overline{Y}_{k}(s))$. Then,
\begin{align}\label{eqn:affine_LPV_SRG}
\operatorname{SRG}(Y(\gamma,s))
\subseteq
\operatorname{SRG}(Y_{0}(s))
\bigoplus_{\forall k} \gamma_{k} \operatorname{SRG}(\overline{Y}_{k}(s))
\end{align}
\end{lemma}
%
\begin{proof}
The result follows directly from the affine structure of $\operatorname{SRG}(Y(\gamma,s))$, Definition~\ref{def:approx_chord}, and Property~\ref{property:sum}.
\end{proof}
%
\begin{example}[Application of Lemma \ref{lemma:SRG_LPV_affine}]\label{example:LPV}
Consider the GFL admittance shown in the Appendix \ref{Appx:GFL_admittance} with parameters of GFL in Table \ref{tab:parameters}. It can be formulated as 
\begin{align*}
    Y_{}(s)={Y_{0}(s)} + i_{d0}{Y_{1}(s)} +  i_{q0}{Y_{2}(s)}.
\end{align*}
The detailed expressions of $Y_{k}(s)$, for $k=0,1,2$, are provided in Appendix~\ref{Appx:GFL_admittance}. Thanks to the affine decomposition, we can compute the SRG of each $Y_k(s)$ independently. First, we plot $\SRG{Y_{0}(s)}$ over the frequency range $\omega \in [10^{-2},10^{3}]$ in Fig.~\ref{fig:SRG_example_Y0}. Moreover, Fig.~\ref{fig:SRG_example_Y1} and Fig.~\ref{fig:SRG_example_Y2} display $\SRG{\overline{Y}_{1}(s)}$ and $\SRG{\overline{Y}_{2}(s)}$ for just four representative frequencies, for clarity purposes.  Note that $\SRG{Y_{1}(s)}$ and $\SRG{Y_{2}(s)}$ correspond only to the orange contour, and the \emph{tight chord} approximation includes all points in gray, satisfying Property~\ref{property:chord}.

\begin{figure}[ht]
    \centering
\begin{subfigure}[b]{0.158\textwidth}
    \centering
    \includegraphics[width=0.95\linewidth]{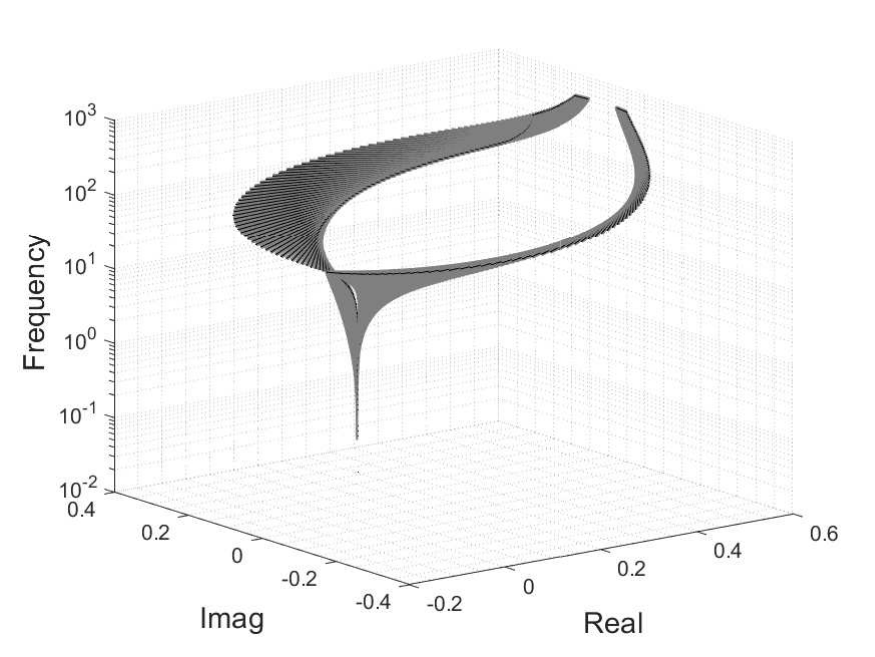}
    \caption{}
    \label{fig:SRG_example_Y0}
\end{subfigure}
\begin{subfigure}[b]{0.158\textwidth}
    \centering
    \includegraphics[width=0.95\linewidth]{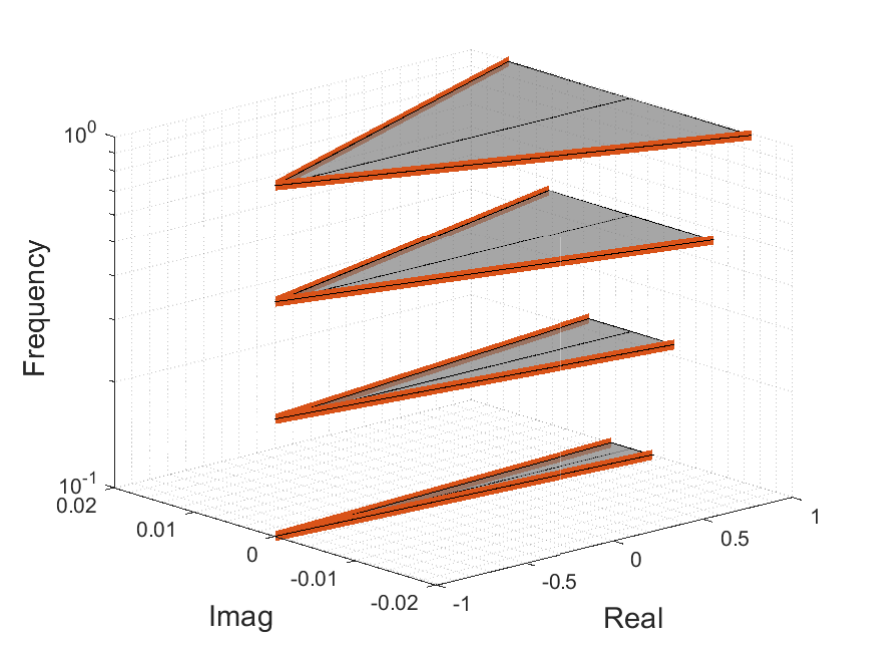}
    \caption{}
    \label{fig:SRG_example_Y1}
\end{subfigure}
\begin{subfigure}[b]{0.158\textwidth}
    \centering
    \includegraphics[width=0.95\linewidth]{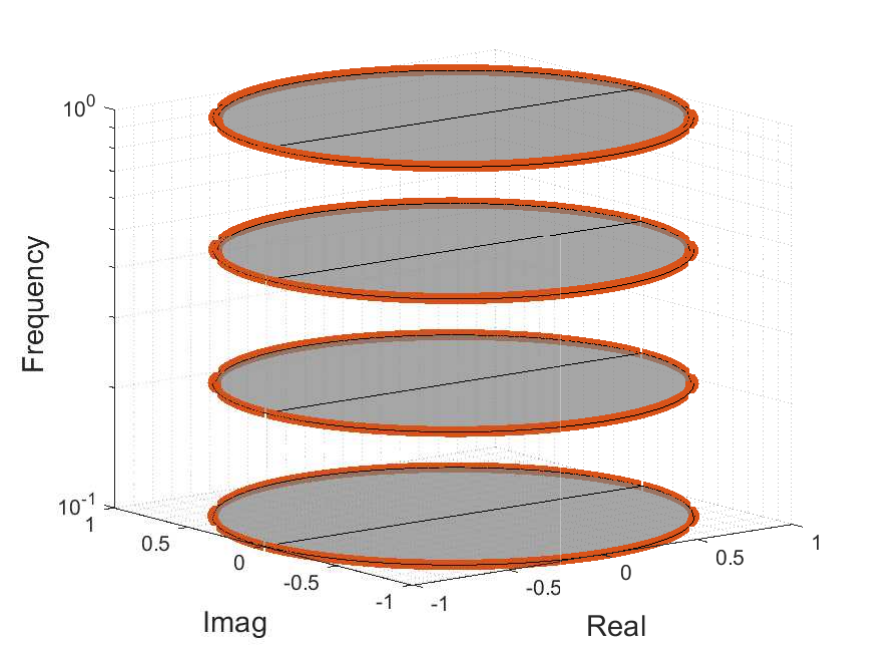}
    \caption{}
    \label{fig:SRG_example_Y2}    
\end{subfigure}
\caption{(a) $\SRG{Y_{0}(s)}$ for $\omega\in[10^{-2},10^3]$ Hz (b) $\SRG{{Y_{1}}(s)}$ in orange  and $\SRG{\overline{Y_{1}}(s)}$ in gray for $\omega\in\{0.1,0.21,0.46,1\}$ Hz (c) $\SRG{{Y_{2}}(s)}$ in orange  and $\SRG{\overline{Y_{2}}(s)}$  in gray for $\omega\in\{0.1,0.21,0.46,1\}$ Hz. }
    \label{fig:SRG_example_1}
    \vspace{-0.5cm}
\end{figure}
\end{example}

For our main theorem, we require two metrics from the frequency-wise SRG projection on the real axis: the center  $ c = \tfrac{(\overline{a}+\underline{a})}{2}$, and the half-width $ \; r = \tfrac{(\overline{a}-\underline{a})}{2},$ where $\overline{a}$ and $\underline{a}$ denote the upper and lower real-axis bounds of the projection.

\begin{example}[Centers and Half-widths]
    From the SRGs in Fig.~\ref{fig:SRG_example_1}, the centers and half-widths of $\operatorname{SRG}(\overline{Y}_{1}(s))$ and $\operatorname{SRG}(\overline{Y}_{2}(s))$ are obtained, as shown in Figs.~\ref{fig:radius_gfl} and~\ref{fig:centers_gfl}. 

\begin{figure}[ht]
    \vspace{-0.3cm}
    \centering
\begin{subfigure}[b]{0.24\textwidth}
    \centering
    \includegraphics[width=1\linewidth]{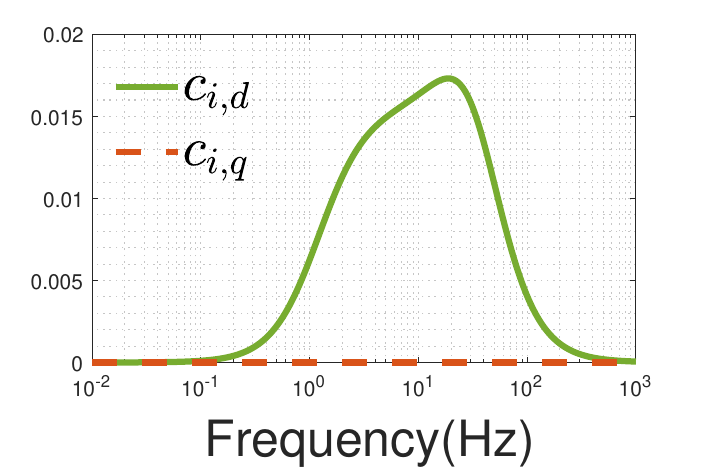}
    \caption{}
    \label{fig:radius_gfl}
\end{subfigure}
\begin{subfigure}[b]{0.24\textwidth}
    \centering
    \includegraphics[width=1\linewidth]{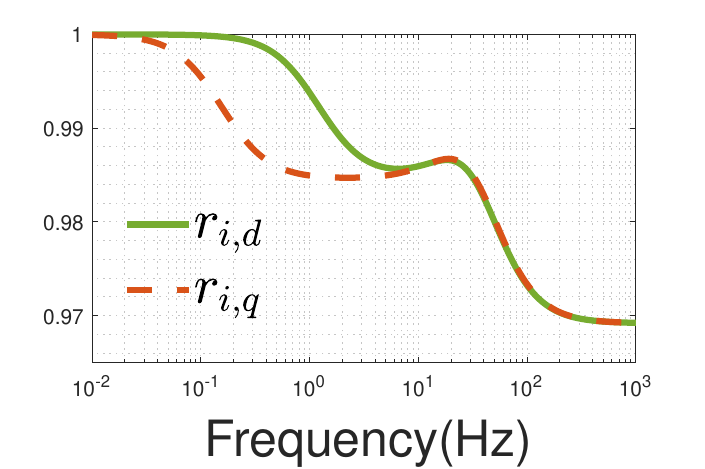}
    \caption{}
    \label{fig:centers_gfl}
\end{subfigure}
    \caption{(a) Centers of $\operatorname{SRG}(\overline{Y}_{1}(s))$ and $\operatorname{SRG}(\overline{Y}_{2}(s))$ (b) Half-widths of $\operatorname{SRG}(\overline{Y}_{1}(s))$ and $\operatorname{SRG}(\overline{Y}_{2}(s))$, for all $\omega\in[10^{-2},10^3]$.}
    \label{fig:centers_radius_example}
    \vspace{-0.1cm}
\end{figure}

\end{example}

 {Furthermore,} to measure the separation between the converter and grid SRGs, we recall the \emph{frequency-wise SRG stability margin}~\cite{Baron2025decentralized}:
\begin{align}\label{eqn:stability_margin}
 \rho(\gamma,\mathrm{j}\omega):=\inf_{\tau\in(0,1]}\operatorname{d}\!\left(\operatorname{SRG}(Y(\gamma,\mathrm{j}\omega)),-\tau\,\operatorname{SRG}(\widehat{Y}_{\mathrm{grid}}(\mathrm{j}\omega))\right). 
\end{align}

A positive margin at every $\omega$,  {i.e.,} $\rho(\gamma,\mathrm{j}\omega) > 0$, guarantees $\mathcal{L}_2$-stability by Theorem~\ref{thm:dec_SRG_S}.  {Note that this margin is dependent on the operation point and therefore on $\gamma$}. {Exploiting properties of the Minkowski sum  {(See Appendix \ref{proof:Thm})}, we derive a theorem that formalizes the connection between admissible operating regions and closed-loop stability, and provides a decentralized stability condition that depends explicitly on the operating point of each converter. We now revisit $n$-converter original problem and state our main theoretical contribution.}

\begin{theorem}\label{thm:main}
Let $Y_{\text{grid}}(s)\in \mathcal{RH}_\infty^{2n \times 2n}$ and  ${\bf \Tilde{Y}}(\Gamma,s) \in \mathcal{RH}_\infty^{2n \times 2n} $ be connected in closed loop as in Fig. \ref{fig:decentralizedfb}, { and both have no poles on $\textup{j}\mathbb{R}\cup \{\infty\}$}. If $ \forall i\in\mathcal{C}$, and $\forall\;\omega \in(0,\infty) $, there exists a set  $\mathcal{S}_i\subseteq \phi_i$ such that, 
$\forall \gamma_i\in\mathcal{S}_i$,
\begin{align}
    \tilde{\rho}_{i,0}(\omega) >\max\Set{0,\bigg| {\sum_{\forall k}}c_{i,k}(\omega) \gamma_{i,k}\bigg|-{\sum_{\forall k}}r_{i,k}(\omega) \abs{\gamma_{i,k}}},
\label{eqn:constraint_rho}
\end{align}
where 
\begin{align}\label{eqn:rho_0}
\tilde{\rho}_{i,0}(\omega):=\inf_{\tau \in (0,1]}\operatorname{d}\left(\operatorname{SRG}(Y_{0}(\mathrm{j}\omega)),-\tau\,\operatorname{SRG}(\widehat{Y}_{\mathrm{grid}}(\mathrm{j}\omega))\right),
\end{align}
then $\rho_i(\gamma_i,\mathrm{j}\omega) > 0$ for all $\omega$, and the system is $\mathcal{L}_2$-stable.
\end{theorem}
\begin{proof}
Proof can be found in Appendix \ref{proof:Thm}.
\end{proof}

Theorem~\ref{thm:main} offers a scalable characterization of each converter stability set, replacing intractable high-dimensional analysis with decentralized conditions that avoid the scaling issues of centralized methods.  {Depending on the converter admittance, this leads either to explicit bounds on  the operating current $i_{dq0}$ (in the GFL case), or to inequalities that define the admissible region of~$\gamma$, which can subsequently be mapped back to constraints on $i_{dq0}$.}

\begin{example}
Consider the converter model of Example~\ref{example:LPV} together with the grid admittance in~\eqref{eqn:Ygrid}, $\tilde{\rho}_{i,0}(\omega)$ can then be computed using~\eqref{eqn:rho_0}, as illustrated in Fig.~\ref{fig:rho_gfl}. 

\begin{figure}[ht]
    \centering
    \includegraphics[width=1\linewidth]{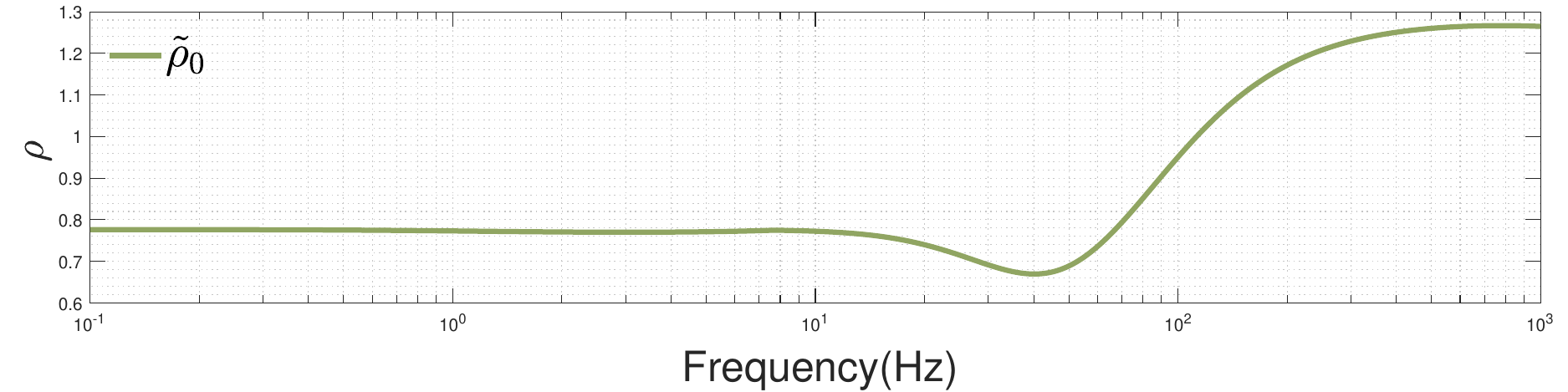}
    \caption{Frequency-wise SRG stability margin $\tilde{\rho}_{i,0}(\omega)$ for GFL.}
    \label{fig:rho_gfl} 
    \vspace{-0.4cm}
\end{figure}

\end{example}

\subsection{Parametric Feasible Set Characterization}
Although Theorem \ref{thm:main} already defines the feasibility space $\mathcal{S}$, we now detail a procedure for obtaining a graphical representation of it. For each converter $i$, the procedure begins by evaluating $\tilde{\rho}_{0}(\omega)$ using~\eqref{eqn:rho_0}. For every $\omega$, we compute the half-width $r_{k}(\omega)$ and center $c_{k}(\omega)$ from $\SRG{\overline{Y}_{k}(s)}$ for all $k$. On top of~\eqref{eqn:constraint_rho}, we introduce operative bounds that are not considered in the admittance transfer function, for instance: 
\begin{align}
    i_{d0}^2+i_{q0}^2 \leq 1.
    \label{eqn:box_bounds}
\end{align}
We define $t = \operatorname{sgn}\!\left(\sum_{\forall k} c_{k}(\omega) \gamma_{k}\right)\in\{\pm 1\}$ and $s_{k} = \operatorname{sgn}(\gamma_{k}) \in \{\pm 1\}$ for all $k$.  In any region where all $s_{k}$, $t$, and $\omega$ are fixed,~\eqref{eqn:constraint_rho} reduces to the following feasibility problem:
\begin{align}\label{lin-ineq}
    \mathcal{F}(s_{k},t,\tilde{\rho}_{0}(\omega)) = 
    \big\{(i_{d0},i_{q0}) \,\big|\,& i_{d0}^2+i_{q0}^2 < 1, \\
    t \Big( { \sum_{\forall k}} c_{k}(\omega)\gamma_{k} (\omega)\Big) \ge &0, \; 
    s_{k} r_{k}(\omega)\gamma_{k}(\omega) \ge 0 \;\forall k, \nonumber \\
     \tilde{\rho}_{0}(\omega) > { \sum_{\forall k}} r_{k}(\omega) &\gamma_{k}(\omega) 
    - { \sum_{\forall k}} c_{k}(\omega)\gamma_{k}(\omega) \big\}. \nonumber
\end{align}

Therefore, the stable operation set is obtained by taking the union over all sign choices 
$s_{k},t \in \{\pm1\}$:
\begin{align}
    \mathcal{F}(\tilde{\rho}_{0}(\omega)) = 
     {\bigcup_{s_{k},t\in\{\pm1\}}} \mathcal{F}(s_{k},t,\tilde{\rho}_{0}(\omega)).
    \label{eqn:feasible_set}
\end{align}
This set is given by the union of at most $2^{n_p+1}$ (possibly empty) sets. Finally, the stability region $\mathcal{S}$ can be found as 
\begin{align*}
    \mathcal{S} = { \bigcap_{\omega\in \mathbb{R}_{>0}}} \mathcal{F}(\tilde{\rho}_{0}(\omega)).
\end{align*}

\begin{remark}[On the $i_{dq0}$ bounds]
 {We assume that GFL and GFM converters can both supply and absorb power from the grid, i.e., $i_{dq0}$ may take positive or negative values. However, in some setups (such as converters interfacing loads or PV panels) $i_{dq0}$ is restricted to only positive or  negative values. In such cases, some $s_{k}$ become fixed, reducing the number of cases to evaluate in $\mathcal{F}(\tilde{\rho}_{i,0}(\omega))$.}
\end{remark}
\subsection{Complexity and Conservatism Analysis}
The proposed method requires less computation than a direct SRG-based stability check, which involves verifying Theorem~\ref{thm:dec_SRG_S} for each converter over the space $\phi$. In our approach,  $\SRG{\overline{Y}_k(s)}\; \forall k$ are computed only once at each frequency using the approach in \cite{pates2021scaled}. This step requires roughly $\mathcal{O}(N_\omega Q n_s^3)$ operations, where $N_\omega$ is the number of frequency samples, $Q$ is the number of angular points used to trace each SRG boundary, and $n_s$ is the system dimension. After this one-time precomputation, the feasible set is obtained by evaluating at most $2^{n_p+1}$ sign pattern combinations. The resulting complexity is therefore $\mathcal{O}(N_\omega Q n^3 + N_\omega 2^{n_p})$. {Notice that the term $2^{n_p}$ is always fairly small; for instance, in the case of GFL, $n_p=2$.}

By contrast, the direct verification approach must check the SRG separation condition for every parameter value $\gamma$ within the region $i_{d0}^2 + i_{q0}^2 \le 1$. Because each evaluation requires recomputing $\SRG{Y_i(\gamma,s)}$ for every $(\gamma,\textup{j}\omega)$ pair, its cost grows as $\mathcal{O}(N_\omega N_\gamma Q n^3)$, where $N_\gamma$ is the number of parameter samples. Since $N_\gamma$ increases rapidly with the desired resolution, this approach becomes computationally expensive even for small system sizes.

While the proposed approach is computationally tractable, it introduces three sources of conservatism. First, the grid SRG is bounded using a frequency-wise disk approximation\cite{Baron2025decentralized}. Second, the chord approximation of each parameter-affine transfer function $Y_k(s)$ adds extra SRG points to the stability assessment. Third, geographical information is not considered.  {While this has minimal impact on the 4-node system case study in Section \ref{subsec:4node}}, it  yields conservative results when grid admittance varies across converter locations. This limitation is inherent to decentralized stability frameworks, which disregard spatial information about converter interconnections\cite{huang2024gain,Haberle2025}.

\section{Case Studies and Validation}

To demonstrate the effectiveness of the proposed methodology,  {we use two validation studies, using a 4-bus system and a modified IEEE 14-node system.} 
\subsection{4-Node System}\label{subsec:4node}

The test configuration comprises a GFM converter and a GFL converter interfaced with an external grid, as illustrated in Fig.~\ref{fig:testsystem}.  Detailed system admittances and parameters are provided in Appendix~\ref{Appx:GFL_admittance} and Appendix~\ref{appendix:GFL_GFM}.

\begin{figure}[ht]
\vspace{-.5cm}
\centering 
    \begin{tikzpicture}[scale=1, every node/.style={transform shape}]
\draw[ thick,color=color2bg] (-0.4,1.5) node[left]{IBR$_{\text{GFM}}$};
\draw[ thick,color=color3bg] (-0.4,2.5) node[left]{IBR$_{\text{GFL}}$};

\draw[ thick,color=black] (6.35,2.63) node[left]{Inf bus};



\draw
(0.5,1)  coordinate(2-node) 
(2-node) --++ (0,0.3) node[above]{}
(2-node)--++(0,0.8) node[right]{2} 
coordinate[pos=0.3](2-r) 
coordinate[pos=0.6](2-rr)
coordinate[pos=0.9](2-rrr)
;
\draw (0,2.5)node[sacdcshape,scale=-.5,color=color3bg,fill=white] (DC1){};

\draw
(0.5,2)  coordinate(3-node) 
(3-node) --++ (0,0.3) node[above]{}
(3-node)--++(0,0.8) node[right]{3} 
coordinate[pos=0.3](3-r) 
coordinate[pos=0.6](3-rr)
coordinate[pos=0.9](3-rrr)
;
\draw (0,1.5)node[sacdcshape,scale=-.5,color=color2bg,fill=white] (DC2){};
\draw
(3,1)  coordinate(4-node) 
(4-node) --++ (0,0.3) node[above]{}
(4-node)--++(0,2) node[left]{4} 
coordinate[pos=0.3](4-r) 
coordinate[pos=0.6](4-rr)
coordinate[pos=0.9](4-rrr)
;

\draw
(5,1.6)  coordinate(5-node) 
(5-node)--++(0,1) node[left]{1} 
coordinate[pos=0.3](5-r) 
coordinate[pos=0.6](5-rr)
coordinate[pos=0.9](5-rrr);

\draw[-stealth](4-r)--++(0.25,0)--++(0,-0.5cm);


\draw
(2-rr)--++(-0.23,0) 
(3-rr)--++(-0.23,0)
(2-rr) to [twoport,bipoles/twoport/width=0.5,a=$Z_2$,bipoles/twoport/height=0.22, fill=white]  (3,1.5)
(3-rr) to [twoport,bipoles/twoport/width=0.5,a=$Z_3$,bipoles/twoport/height=0.22, fill=white]  (3,2.5)
(4-rr) to [twoport,bipoles/twoport/width=0.5,a=$Z_g$,bipoles/twoport/height=0.22, fill=white]  (5-rr)
(5-rr) --++(0.2,0);
\draw
(5.45,2.5) node[below,sin v source]{}
;
\end{tikzpicture}
\caption{GFM and GFL converter test system configuration showing grid topology.}
\label{fig:testsystem}
\end{figure}
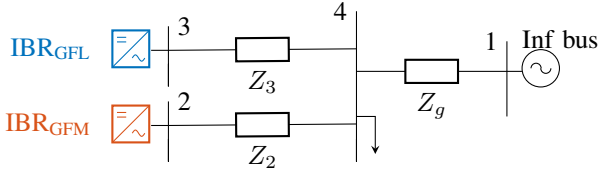

First, the feasible operating region of each converter is characterized, establishing operational boundaries. Representative operating points, both within and beyond these regions, are then selected, and time-domain simulations are carried out to validate the analytical predictions. The external grid impedance is modeled as an inductive load characterized by the admittance matrix \eqref{eqn:Ygrid}, where the aggregate power consumption is $P=0.7$ pu and $Q=0.6$ pu. We include the lines directly into the converter admittance matrices. The frequency-dependent stability margins $\tilde{\rho}_{i,0}(\omega)$ are computed for both converters and presented in Figs.~\ref{fig:rho_gfl} and~\ref{fig:rho_gfm}. The GFL converter has lower stability margins than the GFM converter across the entire frequency spectrum. 

\begin{figure}[ht]
    \vspace{-.4cm}
    \centering
    \includegraphics[width=1\linewidth]{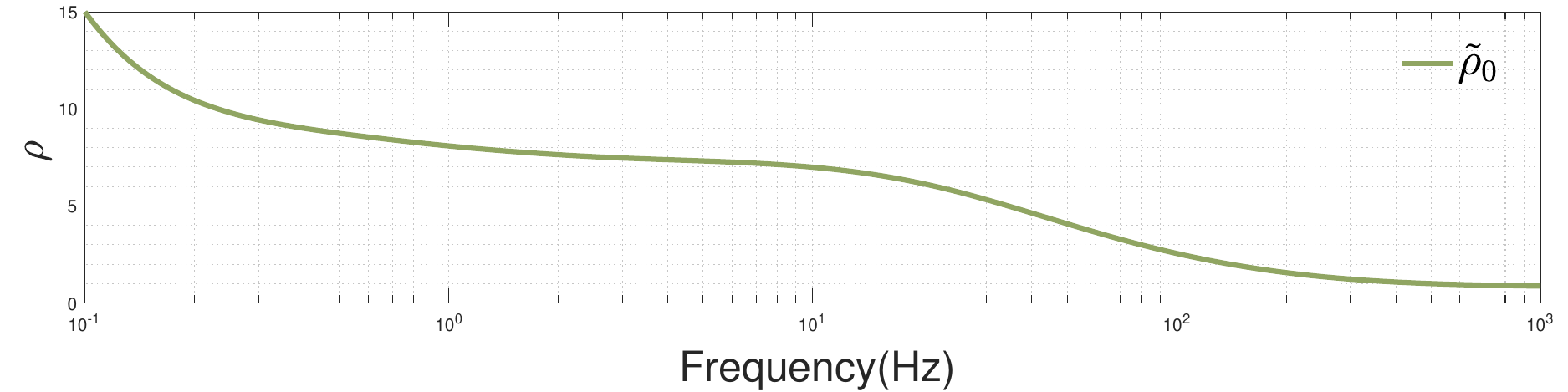}
    \caption{Frequency-wise SRG stability margin $\tilde{\rho}_{i,0}(\omega)$ for GFM.}
    \label{fig:rho_gfm}
\end{figure}

Through the computation of  $r_{k}(\omega)$, and $c_{k}(\omega)$ for each converter, we construct the complete feasibility regions $\mathcal{S}_i$, as depicted in Fig.~\ref{fig:gamma_gfm_gfl}.  {The distinct stability margins $\rho_{0}(\omega)$ induce markedly different operating constraints: the GFL converter admits a narrow diamond-shaped feasible region, while the GFM converter is stable over most of \eqref{eqn:box_bounds}, limited mainly by the maximum deliverable $i_{q0}$, which truncates the circular region at the top and bottom.} 

\begin{figure}[ht]
    \vspace{-.4cm}
    \centering
    \includegraphics[width=0.465\linewidth]{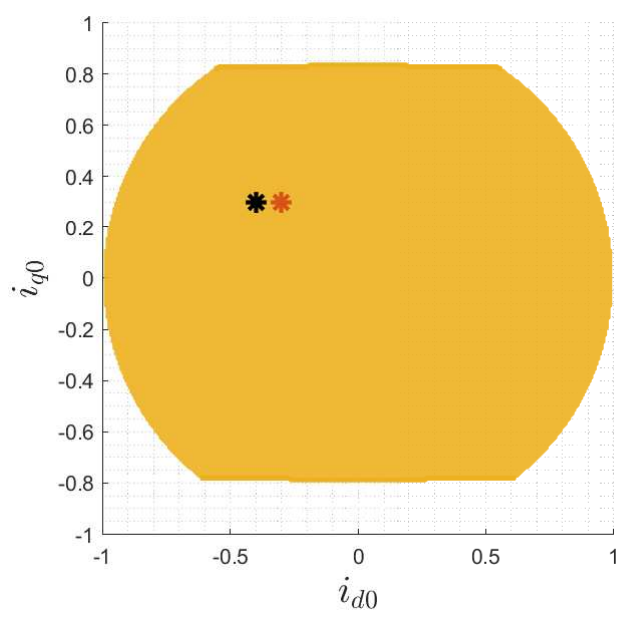}
    \includegraphics[width=0.465\linewidth]{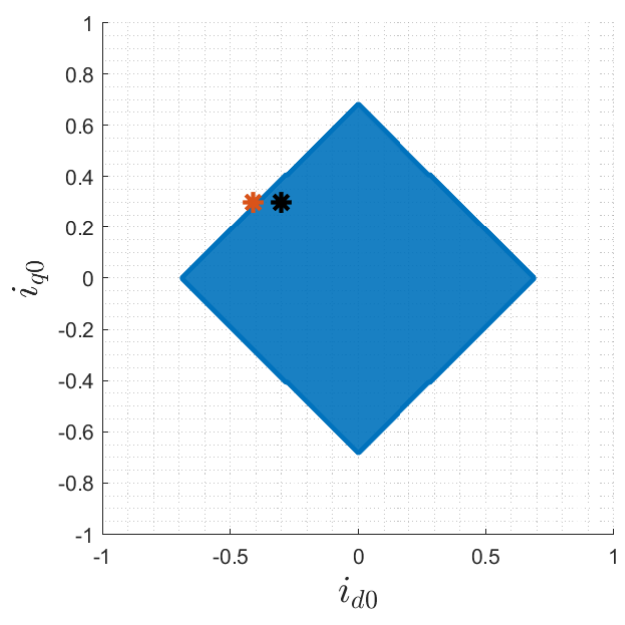}
    \caption{Computed stability operation regions: (left) GFM converter ($\mathcal{S}_{\textup{gfm}}$). (right) GFL converter ($\mathcal{S}_{\textup{gfl}}$). Black asterisk correspond to the stable setpoint and red asterisk correspond to the unstable setpoint.}
    \label{fig:gamma_gfm_gfl}
    \vspace{-.4cm}
\end{figure}
 
To validate the accuracy of the computed stability regions, we conduct time-domain simulations at two operating points. The first validation scenario employs operating points within the predicted stable region: $P_{gfm}=0.4$ pu, $Q_{gfm}=0.3$, and $P_{gfl}=0.3$ pu, $Q_{gfl}=0.3$. As shown in Fig.~\ref{fig:Stable}, the system exhibits transient oscillations that are damped within 0.5 seconds  {for the GFL converter, but 3 seconds for the GFM, both} converging to a stable steady-state condition  {after applying a load step $P=0.06 $ pu in  node 4}. With the reference voltage settings $v_{d}^*=1$ pu and $v_{q}^*=0$ pu, these power setpoints correspond to current references $i_{q,gfm}=0.3$, $i_{d,gfm}=-0.4$ pu and $i_{q,gfl}=0.3$, $i_{d,gfl}=-0.3$ pu, which reside within the blue feasible region in both converters.  {After the perturbation, $i_{q,gfm}$ increases to maintain the voltage at node 2.  The active power required by the disturbance is supplied by the infinite bus.}

\begin{figure}[ht] 
    \centering
    \includegraphics[width=0.91\linewidth]{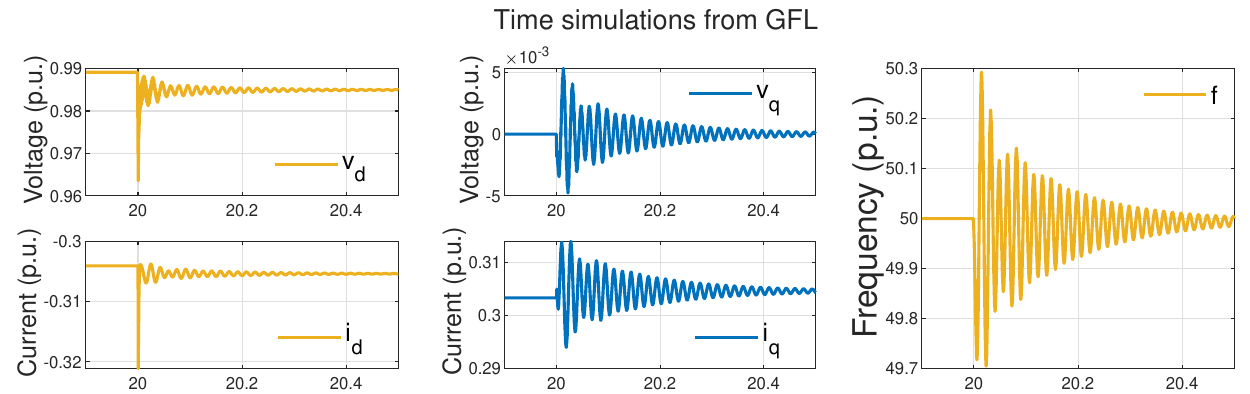}
    \includegraphics[width=0.91\linewidth]{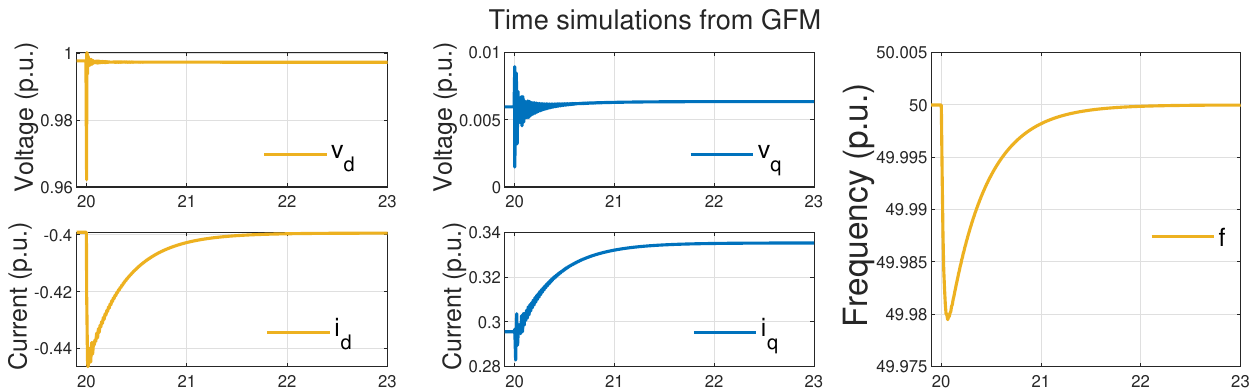}    
    \caption{Time-domain validation of stable operation within predicted feasibility region. Top row corresponds to GFL converter and bottom row to GFM converter}
    \label{fig:Stable}
\end{figure}

The second validation scenario tests an operating point outside the predicted stable region. For setpoints $P_{gfm}=0.3$ pu, $Q_{gfm}=0.3$, and $P_{gfl}=0.4$ pu, $Q_{gfl}=0.3$ pu, the GFL converter exceeds its computed stability boundary. As shown in Fig.~\ref{fig:unstable}, the system becomes unstable, exhibiting divergent oscillations, which confirms our analytical predictions.

\begin{figure}[ht]
    \centering
    \includegraphics[width=0.91\linewidth]{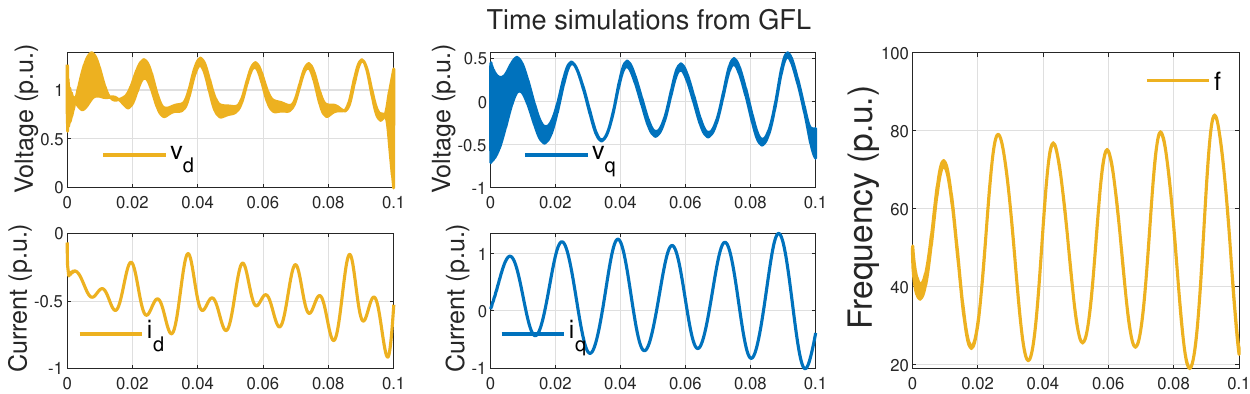}
    \includegraphics[width=0.91\linewidth]{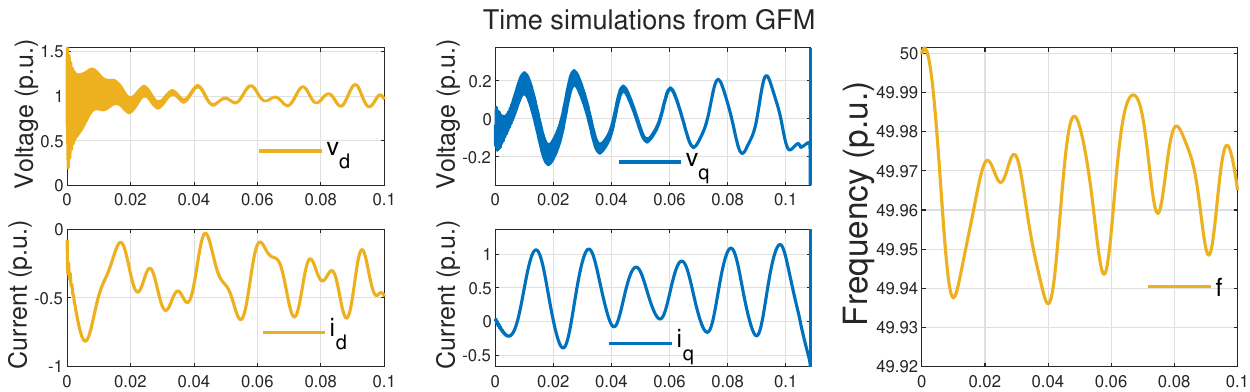}
    \caption{Time-domain validation confirming instability outside predicted feasibility region. Top row corresponds to GFL converter and bottom row to GFM converter}
    \label{fig:unstable}
\vspace{-.5cm}
\end{figure}

\subsection{IEEE 14-Node System}
 {The slack node is located in node 1 and modeled as an infinite bus, while GFL converters are located in nodes 2, 3, 6, and 8 with admittances $Y_c^{N_2}(s),Y_c^{N_3}(s),Y_c^{N_6}(s)$ and $Y_c^{N_8}(s)$, respectively.  {Note that all converters analyzed are GFL, but all have different $f_{\text{pll}}=\{5,10,20,60\}$} as shown in Appendix \ref{appendix:GFL_GFM}.}

\begin{figure}[htb]
\centering
\begin{tikzpicture}[scale=0.75, every node/.style={transform shape}]
\draw[ thick,color=color1bg] (0.5,0.4) to [bend left=90] (-0.5,0.4)node[left]{PCC$_2$};
\draw[ thick,color=color2bg] (7.5,0.4) to [bend right=90] (8.5,0.4)node[right]{PCC$_3$};
\draw[ thick,color=color3bg] (-1.2,-3.25)node[left]{PCC$_6$} to [bend right=90] (-0.6,-3.25);
\draw[ thick,color=color4bg] (7.5,-3.9) to [bend left=90] (8.5,-3.9)node[right]{PCC$_8$};
\draw

(0,-1.5) coordinate(1-node) 
(1-node) --++ (0,-0.15) node[below,sin v source]{}
(1-node)--++(-0.1,0) 
(1-node)--++(1,0) node[right]{1} 
coordinate[pos=0.3](1-r) 
coordinate[pos=0.6](1-rr)
coordinate[pos=0.9](1-rrr)

(0,0)  coordinate(2-node) 
(2-node) --++ (0,0.3) node[above]{}
(2-node)--++(-0.1,0) 
(2-node)--++(1,0) node[right]{2} 
coordinate[pos=0.3](2-r) 
coordinate[pos=0.6](2-rr)
coordinate[pos=0.9](2-rrr)
;
\draw (0,0.5)node[sacdcshape,scale=-.5,color=color1bg,fill=white] (DC1){};
\draw
(8,0)  coordinate(3-node) 
(3-node) --++ (0,0.3) node[above]{}
(3-node)--++(0.1,0) 
(3-node)--++(-1,0) node[left]{3} 
coordinate[pos=0.3](3-r) 
coordinate[pos=0.6](3-rr)
coordinate[pos=0.9](3-rrr)
;
\draw (8,0.5)node[sacdcshape,scale=-.5,color=color2bg,fill=white] (DC2){};
\draw

(8,-1.5)  coordinate(4-node) 
(4-node)--++(0.1,0) 
(4-node)--++(-1,0) node[left]{4} 
coordinate[pos=0.3](4-r) 
coordinate[pos=0.6](4-rr)
coordinate[pos=0.9](4-rrr)


(5,-1.5)  coordinate(5-node) 
(5-node)--++(0.1,0) 
(5-node)--++(-1,0) node[left]{5} 
coordinate[pos=0.3](5-r) 
coordinate[pos=0.6](5-rr)
coordinate[pos=0.9](5-rrr)


(0,-3.5)  coordinate(6-node) 
(6-node)--++(0.1,0) 
(6-node)--++(-1,0) node[left]{6} 
coordinate[pos=0.3](6-r) 
coordinate[pos=0.6](6-rr)
coordinate[pos=0.9](6-rrr)
(6-rrr) --++ (0,0.3) node[above]{}
;
\draw (-0.9,-3)node[sacdcshape,scale=-0.5,color=color3bg,fill=white] (DC3){};
\draw

(8,-2.5)  coordinate(7-node) 
(7-node)--++(0.1,0) 
(7-node)--++(-1,0) node[left]{7} 
coordinate[pos=0.3](7-r) 
coordinate[pos=0.6](7-rr)
coordinate[pos=0.9](7-rrr)

(8,-3.5)  coordinate(8-node) 
(8-node)--++(0.1,0) 
(8-node)--++(-1,0) node[left]{8} 
coordinate[pos=0.3](8-r) 
coordinate[pos=0.6](8-rr)
coordinate[pos=0.9](8-rrr)
(8-node) --++ (0,-0.3) node[below]{}
;
\draw (8,-4)node[sacdcshape,scale=.5,color=color4bg,fill=white] (DC4){};
\draw
(6,-3.5)  coordinate(9-node) 
(9-node)--++(0.1,0) 
(9-node)--++(-1,0) node[left]{9} 
coordinate[pos=0.3](9-r) 
coordinate[pos=0.6](9-rr)
coordinate[pos=0.9](9-rrr)

(4,-3.5)  coordinate(10-node) 
(10-node)--++(0.1,0) 
(10-node)--++(-1,0) node[left]{10} 
coordinate[pos=0.3](10-r) 
coordinate[pos=0.6](10-rr)
coordinate[pos=0.9](10-rrr)


(2,-3.5)  coordinate(11-node) 
(11-node)--++(0.1,0) 
(11-node)--++(-1,0) node[left]{11} 
coordinate[pos=0.3](11-r) 
coordinate[pos=0.6](11-rr)
coordinate[pos=0.9](11-rrr)

(0,-5)  coordinate(12-node) 
(12-node)--++(0.1,0) 
(12-node)--++(-1,0) node[left]{12} 
coordinate[pos=0.3](12-r) 
coordinate[pos=0.6](12-rr)
coordinate[pos=0.9](12-rrr)

(3.5,-5)  coordinate(13-node) 
(13-node)--++(0.1,0) 
(13-node)--++(-1,0) node[left]{13} 
coordinate[pos=0.3](13-r) 
coordinate[pos=0.6](13-rr)
coordinate[pos=0.9](13-rrr)

(7,-5)  coordinate(14-node) 
(14-node)--++(0.1,0) 
(14-node)--++(-1,0) node[left]{14} 
coordinate[pos=0.3](14-r) 
coordinate[pos=0.6](14-rr)
coordinate[pos=0.9](14-rrr)
;
\draw[-stealth](2-node)--++(0,-0.5cm);
\draw[-stealth](3-node)--++(0,-0.5cm);
\draw[-stealth](4-node)--++(0,-0.5cm);
\draw[-stealth](5-node)--++(0,-0.5cm);
\draw[-stealth](6-node)--++(0,-0.5cm);
\draw[-stealth](9-node)--++(0,-0.5cm);
\draw[-stealth](10-node)--++(0,-0.5cm);
\draw[-stealth](11-node)--++(0,-0.5cm);
\draw[-stealth](12-node)--++(0,-0.5cm);
\draw[-stealth](13-node)--++(0,-0.5cm);
\draw[-stealth](14-node)--++(0,-0.5cm);


\draw
(1-r)--(2-r) 
(1-rrr)--++(0,0.25)--++(3.2,0)--++(0,-0.25) 
(2-rr)--++(0,-1)--++(3.75,0)--++(0,-0.5) 
(2-rrr)--++(0,-0.5)--++(6.65,0)--++(0,-1) 
(2-rrr)--++(0,0.25)--++(6.2,0)--++(0,-0.25) 
(3-r)--(4-r) 
(4-rrr)--++(0,0.25)--++(-2.2,0)--++(0,-0.25) 
(5-r)--++(0,-1)--++(-5.25,0)--++(0,-1) 
(6-r)--++(0,0.25)--++(1.5,0)--++(0,-0.25) 
(11-r)--++(0,0.25)--++(1.5,0)--++(0,-0.25) 
(10-r)--++(0,0.25)--++(1.5,0)--++(0,-0.25) 
(8-r)--(7-r) 
(4-rr)--(7-rr) 
(9-node)--++(0,0.5)--++(1.25,0)--++(0,0.5)
(9-r)--++(0,1.5)--++(1.5,0)--++(0,0.5)
(9-rr)--++(0,-0.75)--++(1.5,0)--++(0,-0.75) 
(12-r)--++(0,0.25)--++(3,0)--++(0,-0.25) 
(13-r)--++(0,0.25)--++(3,0)--++(0,-0.25) 
(6-rrr)--(12-rrr)%
(6-rr)--++(0,-0.75)--++(3.5,0)--++(0,-.75)
;
\end{tikzpicture}
\caption{Modified IEEE 14 node system.}
\label{fig:14_node_IEEE}
\end{figure}
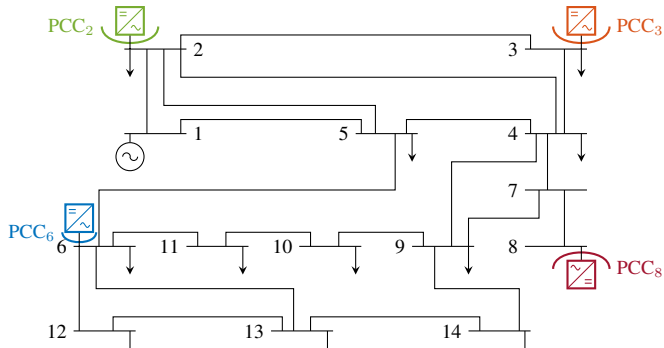

 {It is worth noting that the feasible set expands as $f_{\text{pll}}$ decreases, as illustrated in Fig.~\ref{fig:gamma_14_bus}. This behavior underscores the well-known sensitivity of GFL converters to grid strength, where small-signal stability is strongly conditioned by the PLL bandwidth. As $f_{\text{pll}}$ increases, the converter becomes more sensitive to the external network impedance, effectively requiring a stiffer grid to maintain stable operation. This observation is consistent with findings reported in the literature~\cite{Li_duality_2022,baronprada2026powergrids}.}

\begin{figure}[ht]
    \vspace{-.4cm}
    \centering
        \includegraphics[width=0.7\linewidth]{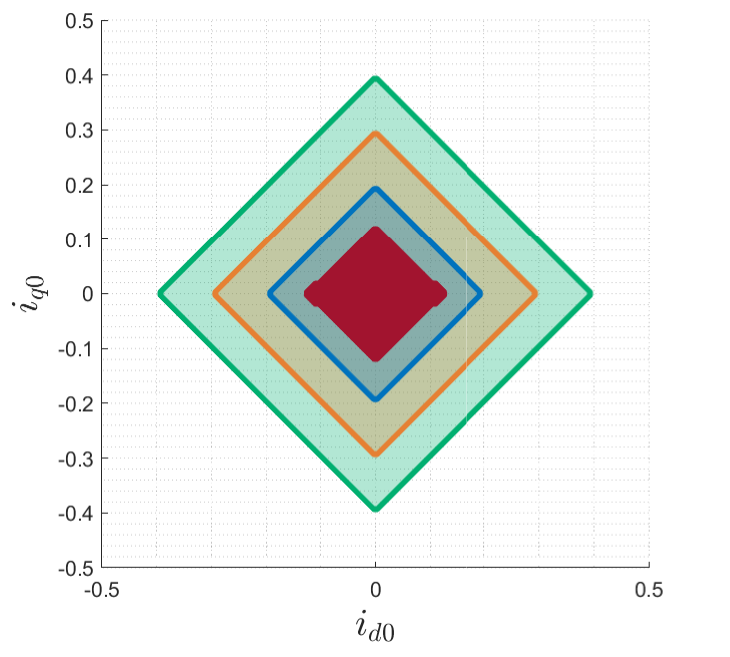}
    \caption{Computed stability operation regions:  GFL converters set estimation $\mathcal{S}({Y_c^{N_2}(s))}$ in green, $\mathcal{S}({Y_c^{N_3}(s))}$ in orange, $\mathcal{S}({Y_c^{N_6}(s))}$ in blue and, $\mathcal{S}({Y_c^{N_8}(s))}$ in red. }
    \label{fig:gamma_14_bus}
    \vspace{-.4cm}
\end{figure}
 {As with most decentralized stability conditions \cite{Haberle2025, huang2024gain}, the proposed framework does not exploit the specific electrical location of each converter within the network, a design choice that enables scalability at the expense of some conservatism. This conservatism is most pronounced in networks with highly nonuniform nodal strength, where location-aware centralized analysis can extract additional margin. In the case where grid strength is approximately uniform across connection points, the two approaches yield closely matching results.}

\section{Conclusions} \label{sec:conclusions}

 {This paper introduces a decentralized framework to identify stability-certified operating points for power systems with GFL and GFM converters. By exploiting the affine dependence of converter admittances on steady-state conditions, small-signal stability can be assessed locally using frequency-wise geometric tests. Each converter computes its admissible region through closed-form linear inequalities, enabling a unified treatment of both converter types. Moreover, the formulation naturally extends to multi-parameter operating sets, where the parameters of interest can be selected by the user according to the application requirements. Numerical results demonstrate the accuracy and computational efficiency of the approach. Future work will address dynamic grid topologies, SRG computation from black-box models, and the incorporation of geographical information to reduce conservatism.}

\bibliographystyle{IEEEtran}
\bibliography{PSCC.bib}

\appendices
\section{Supporting Results and Proofs} \label{proof:Thm}

For completeness, we explicitly characterize the distance between convex sets symmetric about the $x$-axis.

\subsection{Properties of the Distance Function}

We recall several properties of the distance function that are useful for our approach:

 \begin{fact}[Scaling dependence]
         For $\mu \in \mathbb{R}$, $$\operatorname{d}(\mu A, \mu B) = |\mu| \, \operatorname{d}(A,B).$$  
 \end{fact}
\begin{fact}[Minkowski-shift invariance] \label{property:shift}
    The distance is compatible with Minkowski addition. For sets $A,B,C \subset \mathbb{R}^2$,    $$    \operatorname{d}(A, B\oplus C) = \operatorname{d}(A \ominus C, B).$$
\end{fact}
\begin{proof}
By definition, $\operatorname{d}(A, B \oplus C) = \inf_{a \in A,\, b \in B,\, c \in C} \|a - (b + c)\| = \inf_{a \in A,\, b \in B,\, c \in C} \|(a - c) - b\| = \operatorname{d}(A \ominus C, B)$.
\end{proof}
\begin{fact}[Triangle inequality\cite{boyd2004convex}]\label{property:triangle}
    For sets $A,B,C \subset \mathbb{R}^2$, 
    \begin{align*}
        \operatorname{d}(A,B) \ge \abs{\operatorname{d}(A,C) -{\operatorname{d}(B,C)}}.
    \end{align*}
\end{fact}

\begin{lemma}[Distance for Symmetric Convex Sets]\label{thm:distance}
 Consider $\mathbb{R}^2$ with coordinates $(x,y)$. Let $A, B \subset \mathbb{R}^2$ be convex sets symmetric with respect to the $x$-axis, with projections on the $x$-axis $A_x = [\underline{a}, \overline{a}]$ and $B_x = [\underline{b}, \overline{b}]$. Define the centers and half-widths as $c_A = \tfrac{(\underline{a}+\overline{a})}{2}, \, r_A = \tfrac{(\overline{a}-\underline{a})}{2},
c_B = \tfrac{(\underline{b}+\overline{b})}{2}, \, r_B = \tfrac{(\overline{b}-\underline{b})}{2}.$

For real scalars $\mu_A, \mu_B$, let $A' := \mu_A A$ and $B' := \mu_B B$. The centers and half-widths transform as 
\begin{align*}
c_{A'}=\mu_A c_A, \, r_{A'}=|\mu_A| r_A, \,
c_{B'}=\mu_B c_B, \, r_{B'}=|\mu_B| r_B.
\end{align*} 
Then, the distance between $A'$ and $B'$ is
\begin{align*}
    \operatorname{d}(A',B') = \max\Big\{0,\, |\mu_A c_A - \mu_B c_B| - \big(|\mu_A| r_A + |\mu_B| r_B\big)\Big\}.
\end{align*}
\end{lemma}
\begin{proof}
Since $A$ and $B$ are convex and symmetric about the $x$-axis, we have $A_x \times \{0\} \subset A$ and $B_x \times \{0\} \subset B$. Scaling preserves this property, so $A'_x \times \{0\} \subset A'$ and $B'_x \times \{0\} \subset B'$, where $A'_x = \mu_A A_x$ and $B'_x = \mu_B B_x$ denote the projections of the scaled sets. For any $(x_1,y_1) \in A'$ and $(x_2,y_2) \in B'$, the distance satisfies
\begin{align*}
    \operatorname{d}(A', B') = \sqrt{(x_1-x_2)^2+(y_1-y_2)^2} \ge |x_1-x_2|,
\end{align*}
and $|x_1-x_2|= \operatorname{d}(A'_x, B'_x)$, which shows $\operatorname{d}(A',B') \ge \operatorname{d}(A'_x, B'_x).$
Conversely, since $A'_x \times \{0\} \subset A'$ and $B'_x \times \{0\} \subset B'$, for any $x_1 \in A'_x$ and $x_2 \in B'_x$, the points $(x_1, 0) \in A'$ and $(x_2, 0) \in B'$ have distance $\|(x_1,0)-(x_2,0)\| = |x_1 - x_2|$. Therefore,
$\operatorname{d}(A',B') \le \inf\{|x_1 - x_2| : x_1 \in A'_x, x_2 \in B'_x\} = \operatorname{d}(A'_x, B'_x)$. Thus, 
$\operatorname{d}(A',B') = \operatorname{d}(A'_x, B'_x).$

The interval $A'_x$ has endpoints $c_{A'} - r_{A'}$ and $c_{A'} + r_{A'}$, and similarly for $B'_x$. The intervals are disjoint if and only if the gap between their closest endpoints is greater than zero. Without loss of generality, assume $c_{A'} \le c_{B'}$, then the gap is 
$$(c_{B'} - r_{B'}) - (c_{A'} + r_{A'}) = (c_{B'} - c_{A'}) - (r_{A'} + r_{B'}).$$ 
By symmetry, if $c_{B'} \le c_{A'}$, the gap is $(c_{A'} - c_{B'}) - (r_{A'} + r_{B'})$. Thus, in general, and given that $\operatorname{d}(A'_x, B'_x)\geq0$,
\begin{align*}
    \operatorname{d}(A'_x, B'_x) = \max\big\{0,\, |c_{A'} - c_{B'}| - (r_{A'} + r_{B'})\big\}=\operatorname{d}(A', B'),
\end{align*}
which yields the stated formula.
\end{proof} 
\begin{remark}[Extension to Multiple Sets.]
Lemma~\ref{thm:distance} extends to the case where one set is a Minkowski sum of symmetric convex sets.  
Let $A \subset \mathbb{R}^2$ be convex and symmetric about the $x$-axis, and set $B = B_1 \oplus \dots \oplus B_n$, where each $B_j$ is symmetric and convex.  
Define $B' := \sum_{j=1}^n \mu_{B_j} B_j$, the distance between $A'$ and $B'$ reduces to
\begin{align}\label{eqn:multiple}
    \operatorname{d}(A',B') = \max&\Big\{0,\, |\mu_A c_A - { \sum_{j=1}^n} \mu_{B_j} c_{B_j}| \\&- \big(|\mu_A| r_A + { \sum_{j=1}^n} |\mu_{B_j}| r_{B_j}\big)\Big\}\nonumber.
\end{align}
\end{remark}

\subsection{Proof Theorem \ref{thm:main}} 
For brevity, we omit explicit $i$ and $\omega$-dependence in intermediate steps. Recalling \eqref{eqn:stability_margin},
\begin{align}
\rho(\gamma)&=\operatorname{d}\!\Big(\operatorname{SRG}(Y_{}(\gamma)),{ \bigcup_{\tau \in (0,1]}} \tau\,\operatorname{SRG}(\widehat{Y}_{\mathrm{grid}})\Big). \label{eqn:step1_main}
\end{align}
By the affine SRG representation of $\operatorname{SRG}(Y_{}(\gamma))$ in~\eqref{eqn:affine_LPV_SRG}, we can express \eqref{eqn:step1_main} as
\begin{align}
\begin{split}
  \rho(\gamma)\geq\operatorname{d}\Big(&\operatorname{SRG}(Y_{0})\oplus { \sum_{\forall k}} \gamma_{k} \operatorname{SRG}(\overline{Y}_{k}),\\&{ \bigcup_{\tau \in (0,1]}} \tau\,\operatorname{SRG}(\widehat{Y}_{\mathrm{grid}})\Big).   
\end{split}
\label{eqn:step2_main}
\end{align}
Using Property \ref{property:shift}, we can reformulate \eqref{eqn:step2_main} as with 
\begin{align}
\begin{split}
  \rho(\gamma)\geq\operatorname{d}\Big(&\operatorname{SRG}(Y_{0})\ominus{ \bigcup_{\tau \in (0,1]}} \tau\,\operatorname{SRG}(\widehat{Y}_{\mathrm{grid}}), \\&-{ \sum_{\forall k}} \gamma_{k} \operatorname{SRG}(\overline{Y}_{k})\Big).   
\end{split}
\label{eqn:step3_main}
\end{align}

Applying Property~\ref{property:triangle}, with $A=\operatorname{SRG}(Y_{0})\ominus\bigcup_{\tau \in (0,1]} \tau\operatorname{SRG}(\widehat{Y}_{\mathrm{grid}})$, $B=-\sum_{\forall k}\gamma_k \operatorname{SRG}( \overline{Y}_{k})$ and $C = \{0\}$, one arrives at
\begin{align*}
\rho(\gamma)&\ge\Big|\underbrace{\operatorname{d}\!\Big(\operatorname{SRG}(Y_{0})\ominus{ \bigcup_{\tau \in (0,1]}} \tau\,\operatorname{SRG}(\widehat{Y}_{\mathrm{grid}}),0\Big)}_{:=\tilde{\rho}_{0}}
\\
&\quad -\operatorname{d}\Big(-{\sum_{\forall k}}\gamma_k \operatorname{SRG}( \overline{Y}_{k}),\ 0\Big)\Big|.
\end{align*}

Using \eqref{eqn:multiple}, the last term can be expressed as
\begin{align*}
\operatorname{d}\!\Big({ \sum_{\forall k}}\gamma_k &\operatorname{SRG}( \overline{Y}_{k}),\ 0\Big)\\
&=\max\Set{0,\Big| { \sum_{\forall k}}c_{k}(\omega) \gamma_{k}\Big|-{ \sum_{\forall k}}r_{k}(\omega) \abs{\gamma_{k}}}.
\end{align*}
Therefore, a sufficient condition for $\rho_i(\gamma,\mathrm{j}\omega) > 0$ is 
\begin{align*}
    \tilde{\rho}_{0}(\omega) >\max\Set{0,\Big| { \sum_{\forall k}}c_{k}(\omega) \gamma_{k}\Big|-{ \sum_{\forall k}}r_{k}(\omega) \abs{\gamma_{k}}},
\end{align*}
for all $\omega$, and $\mathcal{L}_2$ stability follows from the definition of the stability margin~\eqref{eqn:stability_margin}.


\section{Admittance Models of GFL and GFM Converters}
\label{Appx:GFL_admittance}

This appendix derives the admittance matrices of GFL and GFM converters with classical control loops, following the architecture in Fig.~\ref{fig:control_diagram_GFM_GFL} and the approach of~\cite{Huang2024_Howmany}.
Admittance models are obtained by combining filter dynamics with control equations in the local $dq$-frame and then transforming to global coordinates for system-level analysis. The controllers $G_{\rm CC}(s)$,$G_{\rm PC}(s)$, $G_{\rm VC}(s)$, $M(s)$ and $G_{\rm pll}(s)$ are specified in Table \ref{tab:controllers}. All variables shown in red in Fig.~\ref{fig:control_diagram_GFM_GFL} are measurements with subscript $dq$ and can be expressed as $x_{dq}=[x_d,x_q]^\top$, while those with superscript $*$ denote reference values.

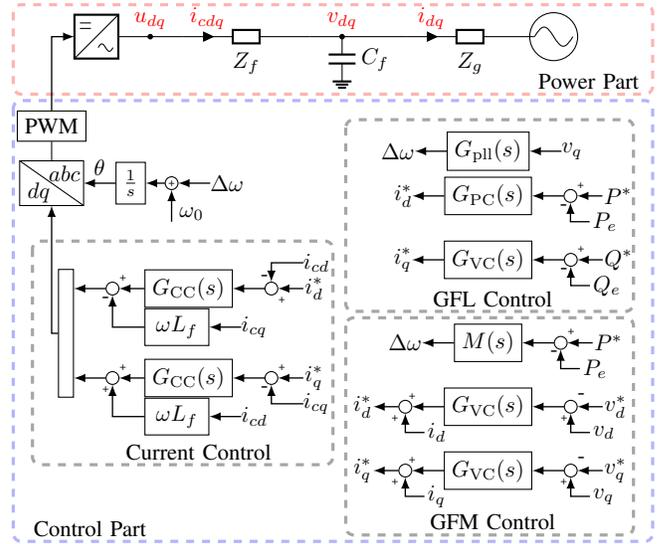
\begin{figure}[ht]
    \centering
\begin{tikzpicture}[scale=0.85, every node/.style={transform shape}, 
sum/.style = {circle, draw, minimum size=4mm, inner sep=0pt}]
\ctikzset{capacitors/height=.3, capacitors/width=0.1}

 \filldraw[color=red!30, fill=red!0,dashed, very thick,rounded corners=3] (-0.5,.5) rectangle (9.5,-0.9);    
 \filldraw[color=blue!30, fill=blue!0,dashed, very thick,rounded corners=3] (-0.5,-1) rectangle (9.5,-7.9);    
\filldraw[color=gray!70, fill=blue!0,dashed, very thick,rounded corners=3] (-0.2,-3.2) rectangle (4.5,-6.7); 
\filldraw[color=gray!70, fill=blue!0,dashed, very thick,rounded corners=3] (4.7,-4.4) rectangle (9.2,-7.8);  
\filldraw[color=gray!70, fill=blue!0,dashed, very thick,rounded corners=3] (4.7,-1.3) rectangle (9.2,-4.3);  
\draw (0.8,0)node[sacdcshape,scale=-.7,color=black,fill=white] (conv){};

\node at (2.4,-6.5) {Current Control};
\node at (7,-7.6) {GFM Control};
\node at (7,-4.1) {GFL Control};
\node at (0.7,-7.7) {Control Part};
\node at (8.5,-.7) {Power Part};

\draw(conv.west) -- ++(1,0) coordinate(mid1)node[midway, above] {\ebr{$u_{dq}$}};
\filldraw[black] (conv.west) ++(0.5,0) circle (1pt) node[anchor=west]{};

\draw (mid1) to[twoport,bipoles/twoport/width=0.3,a=$Z_f$,bipoles/twoport/height=0.15, fill=white,i>^=\ebr{$i_{cdq}$}] ++(2,0) coordinate(mid2);
\draw (mid2) -- ++(0.5,0) coordinate(mid3) node[above] {\ebr{$v_{dq}$}};
\filldraw[black] (mid3) ++(0,0) circle (1pt) node[anchor=west]{};
\draw (mid3) to[scale=1, capacitor, l=$C_f$] ++(0,-0.7) coordinate(c1);
\draw (c1)--++ (0,-0.001) node[scale=0.7,tlground]{};
\draw (mid3) -- ++(1,0) coordinate(mid4);
\draw (mid4) to[twoport,bipoles/twoport/width=0.3,a=$Z_g$,bipoles/twoport/height=0.15, fill=white,i>^=\ebr{$i_{dq}$}] ++(2,0) coordinate(mid5);
\draw (mid5.west) -- ++(0.01,0) coordinate(mid6);
\pic[scale=0.15] at (mid6) {v_sin};
\node[draw,rectangle,minimum width=0.4cm, minimum height=0.25cm] (PWM) at (0.1,-1.4){PWM};
\draw[-latex, line width = .5 pt] (PWM.north) |- (conv.east);
\draw (PWM.south) -- ++(0,-0.3) coordinate(dq_abc) ;
\pic[scale=1] at (dq_abc) {abc_dq};
\draw (dq_abc.east) ++(0.5,-0.35) coordinate(int_thetta) ;
\draw[latex-, line width = .5 pt] (int_thetta) --++ (.5,0) coordinate(int_thetta2) node[midway, above] {$\theta$};
\node[draw,rectangle,minimum width=0.25cm, minimum height=0.4cm,anchor=west] (omega) at (int_thetta2){$\tfrac{1}{s}$};
\draw[latex-, line width = .5 pt] (omega.east) -- ++(0.3,0) coordinate(omega_2) ;
\node[sum,anchor=west,scale=0.5] (sum1) at (omega_2) {+};
\draw[latex-, line width = .5 pt] (sum1.east) -- ++(0.5,0) coordinate(omega_3) node[near end, right] {$\Delta\omega$} ;
\draw[latex-, line width = .5 pt] (sum1.south) -- ++(0,-0.5) coordinate(omega_21) node[near end, right] {$\omega_0$} ;

\draw(dq_abc) ++(0,-0.7) coordinate(out_cc) ;
\draw[latex-, line width = .5 pt] (out_cc) --++(0,-2)-- ++(0.1,0) coordinate(v_ref) ;
\node[draw,rectangle,minimum width=0.01cm, minimum height=2cm,anchor=west] (vector) at (v_ref){};
\draw[latex-, line width = .5 pt] (vector.east)++(0,0.7) --++(0.5,0) coordinate(vd_ref) node[pos=.995,below,sloped] () {-};
\node[sum,anchor=west,scale=0.5] (sum2) at (vd_ref) {};
\draw[latex-, line width = .5 pt] (sum2.south) -- ++(0,-0.5)--++(0.5,0) coordinate(wLf_1) ;
\node[draw,rectangle,minimum width=1cm, minimum height=0.5cm,anchor=west] (wlf_icq) at (wLf_1){$\omega L_f$};
\draw[latex-, line width = .5 pt] (wlf_icq.east) --++(0.5,0)node[near end, right] {$i_{cq}$};
\draw[latex-, line width = .5 pt] (sum2) --++(0.5,0) coordinate(vd_ref1) node[pos=.15,above,sloped] () {\tiny+};
\node[draw,rectangle,minimum width=1cm, minimum height=0.5cm,anchor=west] (cc_d) at (vd_ref1){$G_{\rm CC}(s)$};
\draw[latex-, line width = .5 pt] (cc_d.east)--++(0.5,0) coordinate(vd_ref2) node[pos=.995,above,sloped] () {-};
\node[sum,anchor=west,scale=0.5] (sum3) at (vd_ref2) {} ;
\draw[latex-, line width = .5 pt] (sum3) --++(0.5,0) coordinate(vd_ref2) node[near end, right] {$i_{d}^*$} node[pos=.15,below,sloped] () {\tiny +};
\draw[latex-, line width = .5 pt] (sum3.north) -- ++(0,0.3) --++(0.5,0) coordinate(icd) node[near end, right] {$i_{cd}$} ;
\draw[latex-, line width = .5 pt] (vector.east)++(0,-.7) --++(0.5,0) coordinate(vd_ref)  node[pos=.995,below,sloped] () {\tiny+} ;
\node[sum,anchor=west,scale=0.5] (sum2) at (vd_ref) {};
\draw[latex-, line width = .5 pt] (sum2.south) -- ++(0,-0.5)--++(0.5,0) coordinate(wLf_1) ;
\node[draw,rectangle,minimum width=1cm, minimum height=0.5cm,anchor=west] (wlf_icq) at (wLf_1){$\omega L_f$};
\draw[latex-, line width = .5 pt] (wlf_icq.east) --++(0.5,0)node[near end, right] {$i_{cd}$};
\draw[latex-, line width = .5 pt] (sum2) --++(0.5,0) coordinate(vd_ref1) node[pos=.15,above,sloped] () {\tiny+};
\node[draw,rectangle,minimum width=1cm, minimum height=0.5cm,anchor=west] (cc_d) at (vd_ref1){$G_{\rm CC}(s)$};
\draw[latex-, line width = .5 pt] (cc_d.east)--++(0.5,0) coordinate(vd_ref2) node[pos=.995,below,sloped] () {-};
\node[sum,anchor=west,scale=0.5] (sum3) at (vd_ref2) {};
\draw[latex-, line width = .5 pt] (sum3) --++(0.5,0) coordinate(vd_ref2) node[near end, right] {$i_{q}^*$}node[pos=.15,above,sloped] () {\tiny +};
\draw[latex-, line width = .5 pt] (sum3.south) -- ++(0,-0.3)--++(0.5,0) coordinate(icd) node[near end, right] {$i_{cq}$};

\draw(omega_3) ++(3.7,0.5) coordinate(GFL) ;
\node[draw,rectangle,minimum width=1cm, minimum height=0.5cm,anchor=west] (G_pll) at (GFL){$G_{\rm pll}(s)$};
\draw[latex-, line width = .5 pt] (G_pll.east)--++(0.5,0) coordinate(gpll_input) node[near end, right] {$v_{q}$};;
\draw[-latex, line width = .5 pt] (G_pll.west)--++(-0.5,0) coordinate(gpll_output) node[near end, left] {$\Delta \omega$} ;

\draw(G_pll) ++(0,-.7) coordinate(G_pll2) ;
\node[draw,rectangle,minimum width=1cm, minimum height=0.5cm] (G_pc) at (G_pll2){$G_{\rm PC}(s)$};
\draw[latex-, line width = .5 pt] (G_pc.east)--++(0.5,0) coordinate(gpc_input) node[pos=.995,below,sloped] () {-};
\node[sum,anchor=west,scale=0.5] (sum3) at (gpc_input) {};
\draw[latex-, line width = .5 pt] (sum3.south) -- ++(0,-0.3)--++(0.3,0) coordinate(icd) node[near end, right] {$P_{e}$};
\draw[latex-, line width = .5 pt] (sum3) --++(0.5,0) coordinate(vd_ref2) node[near end, right] {$P^{*}$} node[pos=.15,above,sloped] () {\tiny+};
\draw[-latex, line width = .5 pt] (G_pc.west)--++(-0.5,0) coordinate(gpc_output) node[near end, left] {$i_{d}^*$};

\draw(G_pc) ++(0,-1) coordinate(G_pc2) ;
\node[draw,rectangle,minimum width=1cm, minimum height=0.5cm] (G_vc) at (G_pc2){$G_{\rm VC}(s)$};
\draw[latex-, line width = .5 pt] (G_vc.east)--++(0.5,0) coordinate(gvc_input) node[pos=.995,below,sloped] () {-};
\node[sum,anchor=west,scale=0.5] (sum3) at (gvc_input) {};
\draw[latex-, line width = .5 pt] (sum3.south) -- ++(0,-0.3)--++(0.3,0) coordinate(icd) node[near end, right] {$Q_{e}$};
\draw[latex-, line width = .5 pt] (sum3) --++(0.5,0) coordinate(vd_ref2) node[near end, right] {$Q^*$}  node[pos=.15,above,sloped] () {\tiny+};
\draw[-latex, line width = .5 pt] (G_vc.west)--++(-0.5,0) coordinate(gvc_output) node[near end, left] {$i_{q}^*$};

\draw(G_pc2) ++(0,-1.3) coordinate(GFL) ;
\node[draw,rectangle,minimum width=1cm, minimum height=0.5cm] (G_pll) at (GFL){$M(s)$} ;
\draw[latex-, line width = .5 pt] (G_pll.east)--++(0.5,0) coordinate(gpll_input)  node[pos=.995,below,sloped] () {-};
\node[sum,anchor=west,scale=0.5] (sum3) at (gpll_input) {};
\draw[latex-, line width = .5 pt] (sum3.south) -- ++(0,-0.3)--++(0.3,0) coordinate(icd) node[near end, right] {$P_{e}$};
\draw[latex-, line width = .5 pt] (sum3) --++(0.5,0) coordinate(vd_ref2) node[near end, right] {$P^{*}$} node[pos=.15,above,sloped] () {\tiny+};
\draw[-latex, line width = .5 pt] (G_pll.west)--++(-0.5,0) coordinate(gpll_output) node[near end, left] {$\Delta \omega$};

\draw(G_pll) ++(0,-1) coordinate(G_pll2) ;
\node[draw,rectangle,minimum width=1cm, minimum height=0.5cm] (G_pc) at (G_pll2){$G_{\rm VC}(s)$};
\draw[latex-, line width = .5 pt] (G_pc.east)--++(0.5,0) coordinate(gpc_input) node[pos=.995,below,sloped] () {\tiny +};
\node[sum,anchor=west,scale=0.5] (sum3) at (gpc_input) {};
\draw[latex-, line width = .5 pt] (sum3.south) -- ++(0,-0.3)--++(0.3,0) coordinate(icd) node[near end, right] {$v_{d}$};
\draw[latex-, line width = .5 pt] (sum3) --++(0.5,0) coordinate(vd_ref2) node[near end, right] {$v_d^{*}$} node[pos=.15,above,sloped] () {-};
\draw[-latex, line width = .5 pt] (G_pc.west)--++(-0.5,0) coordinate(gpc_output) node[pos=.995,above,sloped] () {\tiny+};
\node[sum,anchor=east,scale=0.5] (sum3) at (gpc_output) {};
\draw[latex-, line width = .5 pt] (sum3.south) -- ++(0,-0.3)--++(0.3,0) coordinate(icd) node[near end, right] {$i_{d}$};
\draw[-latex, line width = .5 pt] (sum3) --++(-0.5,0) coordinate(vd_ref2) node[near end, left] {$i_{d}^*$} node[pos=.15,below,sloped] () {\tiny+};

\draw(G_pc) ++(0,-1) coordinate(G_pc2) ;
\node[draw,rectangle,minimum width=1cm, minimum height=0.5cm] (G_vc) at (G_pc2){$G_{\rm VC}(s)$};
\draw[latex-, line width = .5 pt] (G_vc.east)--++(0.5,0) coordinate(gvc_input) node[pos=.995,below,sloped] () {\tiny +};
\node[sum,anchor=west,scale=0.5] (sum3) at (gvc_input) {};
\draw[latex-, line width = .5 pt] (sum3.south) -- ++(0,-0.3)--++(0.3,0) coordinate(icd) node[near end, right] {$v_{q}$};
\draw[latex-, line width = .5 pt] (sum3) --++(0.5,0) coordinate(vd_ref2) node[near end, right] {$v_q^{*}$} node[pos=.15,above,sloped] () {-};
\draw[-latex, line width = .5 pt] (G_vc.west)--++(-0.5,0) coordinate(gvc_output) node[pos=.995,above,sloped] () {\tiny+};
\node[sum,anchor=east,scale=0.5] (sum3) at (gvc_output) {};
\draw[latex-, line width = .5 pt] (sum3.south) -- ++(0,-0.3)--++(0.3,0) coordinate(icd) node[near end, right] {$i_{q}$};
\draw[-latex, line width = .5 pt] (sum3) --++(-0.5,0) coordinate(vd_ref2) node[near end, left] {$i_{q}^*$}node[pos=.15,below,sloped] () {\tiny +};


\end{tikzpicture}

    \caption{GFL and GFM control structures, where $Z_f=R_f+\textup{j}\omega L_f$. }
    \label{fig:control_diagram_GFM_GFL}
    \vspace{-0.6cm}
\end{figure}

\subsection{Admittance Model of GFL Converters}
We derive the GFL (PLL-based) admittance matrix in Fig.~\ref{fig:control_diagram_GFM_GFL}, modeled in the controller’s rotating $dq$-frame.  The filter dynamics are 
\begin{align}
{  u_{dq}^*} -   v_{dq} = (sL_f + \textup{j}\omega L_f+R_f)  i_{cdq}, \label{eqn:GFL_Lf}
\end{align}
and the current controller is
\begin{align}
{  u_{dq}^*} = {G}_{\rm CC}(s)(  i_{dq}^{*}-  i_{cdq}) + \textup{j}\omega L_f  i_{cdq} . \label{eqn:GFL_current_control}
\end{align}
Combining \eqref{eqn:GFL_Lf}–\eqref{eqn:GFL_current_control} gives $$G_1(s)  i_{dq}^{*} - G_2(s)  v_{dq} =   i_{cdq} ,$$
with
\begin{align*}
G_1(s)=\tfrac{{G}_{\rm CC}(s)}{sL_f+R_f+{G}_{\rm CC}(s)}, \; 
G_2(s)=\tfrac{1}{sL_f+R_f+{G}_{\rm CC}(s)}. 
\end{align*}

Outer power/voltage loops set the current reference as
\begin{align*}
i_{cd}^{*}={G}_{\rm PC}(s)(P^{*}-P_E), \;
i_{cq}^{*}={G}_{\rm VC}(s)(Q^{*}-Q_E), 
\end{align*}
with $P_E=v_{d} i_{cd}+v_{q} i_{cq}$ and $Q_E=v_{d} i_{cq}-v_{q} i_{cd}$.  Linearizing around an operation point $(i_{d0},i_{q0},v_{d0},v_{q0})$  yields
\begin{align}
-\begin{bmatrix}\Delta i_{cd}\\ \Delta i_{cq}\end{bmatrix}
= \begin{bmatrix} Y_{11}(s) & Y_{12}(s) \\ Y_{21}(s) & Y_{22}(s) \end{bmatrix}
\begin{bmatrix}\Delta v_{d} \\ \Delta v_{q}\end{bmatrix}. \label{eqn:GFL_admittance}
\end{align}
where, $v_{q0}=0$, $v_{d0}=1$, and
\begin{align*}
    Y_{11}&=\tfrac{G_1(s)G_{\rm PC}(s)i_{d0}+G_2(s)}{1+G_1(s)G_{\rm PC}(s)}, &
    Y_{12}&=\tfrac{G_1(s)G_{\rm PC}(s)i_{q0}}{1+G_1(s)G_{\rm PC}(s)},\\
    Y_{21}&=\tfrac{G_1(s)G_{\rm VC}(s)i_{q0}}{1+G_1(s)G_{\rm VC}(s)}, &
    Y_{22}&=\tfrac{-G_1(s)G_{\rm VC}(s)i_{d0}+G_2(s)}{1+G_1(s)G_{\rm VC}(s)}.
\end{align*}
To refer \eqref{eqn:GFL_admittance} to the global $dq$-frame, we consider the PLL dynamics as
\begin{align}
\Delta\theta = \tfrac{{G}_{\rm pll}(s)}{s}\Delta v_{q} , \label{eqn:GFL_PLL}
\end{align}
which couples the converter to the global frame.  The admittance in \eqref{eqn:GFL_admittance} is expressed in the local $dq$-frame synchronized by the PLL.  To derive a system model valid in a common reference, we transform to a global frame rotating at the fixed angular frequency $\omega_0$.  Let ${\theta_i}=\theta-\theta_G$ denote the phase difference between local and global coordinates.  The voltage and current perturbations transform as
\begin{align}
\begin{bmatrix}\Delta v'_{d}\\ \Delta v'_{q}\end{bmatrix}
= {J(\theta_i)}\!\left(\begin{bmatrix}\Delta v_{d}\\ \Delta v_{q}\end{bmatrix}
+ \begin{bmatrix}-v_{q0}\\ v_{d0}\end{bmatrix}\Delta\delta \right), \label{eqn:V_global}
\end{align}
\begin{align}
\begin{bmatrix}\Delta i'_{cd}\\ \Delta i'_{cq}\end{bmatrix}
= {J(\theta_i)}\!\left(\begin{bmatrix}\Delta i_{cd}\\ \Delta i_{cq}\end{bmatrix}
+ \begin{bmatrix}-i_{q0}\\ i_{d0}\end{bmatrix}\Delta\delta \right), \label{eqn:I_global}
\end{align}
where  $\Delta\delta=\Delta\theta$ from the PLL dynamics \eqref{eqn:GFL_PLL}.   Combining \eqref{eqn:V_global}–\eqref{eqn:I_global} with \eqref{eqn:GFL_PLL}, we arrive to the GFL admittance in the global $dq$-frame.
\begin{align*}
\text{-}\begin{bmatrix}\Delta i'_{cd}\\ \Delta i'_{cq}\end{bmatrix}
= {J(\theta_i)}
\begin{bmatrix}
Y_{11}(s) & \tfrac{sY_{12}(s)+{G}_{\rm pll}(s)i_{q0}}{s+{G}_{\rm pll}(s)} \\
Y_{21}(s) & \tfrac{sY_{22}(s)-{G}_{\rm pll}(s)i_{d0}}{s+{G}_{\rm pll}(s)}
\end{bmatrix}
{J(\text{-}\theta_i)}\begin{bmatrix}\Delta v'_{d}\\ \Delta v'_{q}\end{bmatrix}. 
\end{align*}
Observe that $Y_{11}(s)$ and $Y_{22}(s)$ decompose into terms involving $i_{d0}$ and terms independent of any component. In contrast, $Y_{12}(s)$ and $Y_{21}(s)$ depend exclusively on $i_{q0}$, as shown below, which can be rewritten as $Y_0(s)$, $Y_1(s)$ and $Y_2(s)$.
\begin{align*}
Y_0\hspace{-0.1cm}&=\hspace{-0.15cm}\begin{bmatrix}
    \tfrac{G_2(s)}{1+G_1(s)G_{\rm PC}(s)}     & 0 \\ 0      & \tfrac{sG_2(s)}{(1+G_1(s)G_{\rm VC}(s))(s+{G}_{\rm pll})}
    \end{bmatrix}\\
Y_1\hspace{-0.1cm}&=\hspace{-0.15cm}\begin{bmatrix}
    \tfrac{G_1(s)G_{\rm PC}}{1+G_1(s)G_{\rm PC}(s)}     & 0 \\ 0      & \tfrac{-s(G_1(s)G_{\rm VC}(s))}{(1+G_1(s)G_{\rm VC}(s))(s+{G}_{\rm pll})}-\tfrac{{G}_{\rm pll}}{(s+{G}_{\rm pll})}
    \end{bmatrix}\\
Y_2\hspace{-0.1cm}&=\hspace{-0.15cm}\begin{bmatrix}
        0 &\tfrac{s(G_1(s)G_{\rm PC}(s))}{(s+{G}_{\rm pll}(s))(1+G_1(s)G_{\rm PC}(s))} +\tfrac{{G}_{\rm pll}(s)}{s+{G}_{\rm pll}(s)}\\ 
        \tfrac{G_1(s)G_{\rm VC}(s)}{1+G_1(s)G_{\rm VC}(s)} & 0
    \end{bmatrix}\hspace{-0.15cm}.    
\end{align*}
\subsection{Admittance Model of GFM Converters}
\label{Appx:GFc_Admittance}
For the GFM converter in Fig.~\ref{fig:control_diagram_GFM_GFL}, consider \eqref {eqn:GFL_Lf} and
\begin{align}
  i_{cdq} -   i_{dq} &= (sC_f + \textup{j}\omega C_f)\,  v_{dq} .
\label{eqn:GFM_filters}
\end{align}

The voltage control loop is
\begin{align}
  i_{cdq}^* \;=\; G_{\rm VC}(s)\,(  v_{dq}^{*}-  v_{dq}) + \textup{j}\omega C_f  v_{dq} +   i_{dq} .
\label{eqn:GFM_voltage_control}
\end{align}
Finally, the synchronization loop is taken as
\begin{align}
\Delta \delta \;=\; M(s)\,\Delta P_E,
\quad
M(s) \;:=\; \tfrac{1}{J s^{2}+D s}.
\label{eqn:GFM_freq_loop}
\end{align}

By linearizing around an operation point , with $v_{q0}=0$, and eliminating $  u_{dq}^*$ and $  i_{cdq}$ from \eqref{eqn:GFL_Lf},  \eqref{eqn:GFL_current_control}, \eqref{eqn:GFM_filters}, and \eqref{eqn:GFM_voltage_control} yields a local admittance $Y_0(s)$ such that, in the local $dq$ frame,
\begin{align*}
\text{-}
\begin{bmatrix}\Delta i_{cd}\\ \Delta i_{cq}\end{bmatrix}
=
\begin{bmatrix}\tfrac{G_1G_{\rm VC}+sC_fG_1+G_2}{1-G_1}&0\\0&\tfrac{G_1G_{\rm VC}+sC_fG_1+G_2}{1-G_1}\end{bmatrix}
\begin{bmatrix}\Delta v_{d}\\ \Delta v_{q}\end{bmatrix},
\end{align*}

The active power linearization around $(v_{d0},0,i_{d0},i_{q0})$ is
\begin{align}
\Delta P_E \approx v_{d0}\,\Delta i_{cd} + i_{d0}\,\Delta v_d + i_{q0}\,\Delta v_q .
\label{eq:S4}
\end{align}

Substituting  \eqref{eq:S4} into \eqref{eqn:GFM_freq_loop}  and then into \eqref{eqn:V_global}, gives
\begin{align}
\Delta\delta=\tfrac{-\,(i_{d0} - v_{d0} Y_0(s))\,\Delta v'_d \;\;-\;\; i_{q0}\,\Delta v'_q}{\,M(s)-i_{q0}v_{d0}\,}.
\label{eq:S6_correct}
\end{align}

Finally, substituting \eqref{eq:S6_correct} into \eqref{eqn:I_global}:
\begin{align*}
\Delta i'_{cd}&= -Y_0\,\Delta v'_d - i_{q0}\,\Delta\delta, \\
\Delta i'_{cq}&= -Y_0\,\Delta v'_q + \big(Y_0 v_{d0}+i_{d0}\big)\,\Delta\delta. 
\end{align*}
Then, the admittance in global coordinates is defined by
\begin{align*}
-
\begin{bmatrix}\Delta i'_{cd}\\ \Delta i'_{cq}\end{bmatrix}
=
{J(\theta_i)}Y_{\rm gfm}(s){J(-\theta_i)}
\begin{bmatrix}\Delta v'_d\\ \Delta v'_q\end{bmatrix},
\end{align*}
with
\begin{align}
Y_{\rm gfm}(s)=
\begin{bmatrix}
  Y_0 -  \tfrac{i_{q0}(i_{d0} - v_{d0} Y_0)}{M(s)-i_{q0}v_{d0}}
&
 - \tfrac{i_{q0}^{2}}{M(s)-i_{q0}v_{d0}}
\\
 \tfrac{(Y_0 v_{d0}+i_{d0})(i_{d0} - v_{d0} Y_0)}{M(s)-i_{q0}v_{d0}}
&
 Y_0 +  \tfrac{(Y_0 v_{d0}+i_{d0})i_{q0}}{M(s)-i_{q0}v_{d0}}
\end{bmatrix}
\label{eq:Ygfm_exact}
\end{align}
$Y_{\rm gfm}(s)$ is not affine LPV in $(i_{d0},i_{q0})$ due to the term $(M(s)-i_{q0}v_{d0})^{-1}$. When $i_{q0}=0$, the coupling vanishes and an LPV form in $i_{d0}$ is attainable. For the general case with $i_{q0}\neq 0$, we approximate the factor $ \big(M(s)-i_{q0}v_{d0}\big)^{-1} $ in the denominator via a shifted geometric series,  as follows
\begin{align}
    \tfrac{1}{M(s)-i_{q0}v_{d0}}\
= {\sum_{k=0}^{N_i}}\tfrac{(i_{q0}v_{d0}+\alpha_i)^{k}}{\big(M(s)+\alpha_i\big)^{k+1}}=:G_i \label{eqn:approx}
\end{align}
which replacing \eqref{eqn:approx} in \eqref{eq:Ygfm_exact} leads to
\begin{align*}
Y_{\rm gfm}(s) \approx\hspace{7.5cm}\nonumber\\
\begin{bmatrix}
 Y_0 - i_{q0}\,\big( i_{d0} - \! v_{d0} Y_0\big)\,G_i &
 -\,i_{q0}^{2}\,G_i \\
  \big(Y_0 v_{d0} + i_{d0}\big)\big( i_{d0} -  v_{d0} Y_0\big)G_i &
 Y_0 + \big(Y_0 v_{d0} + i_{d0}\big)\,i_{q0}\,G_i 
\end{bmatrix}.
\end{align*}

To achieve a better approximation in our example for the parameters in Table \ref{tab:parameters}, we chose $(\alpha_i, N_i)=(-15,\,1)$. The LPV parameter derivation is omitted for brevity.


%
\section{Controller Definitions and Parameters} \label{appendix:GFL_GFM}
The controllers used in both GFL and GFM converters are summarized in Table~\ref{tab:controllers}.  
\begin{table}[ht]
\centering
\caption{Controller transfer functions.}
\label{tab:controllers}
\begin{tabular}{ll}
\toprule
\textbf{Controller} & \textbf{Transfer Function} \\
\midrule
Current control (GFL, GFM) & ${G}_{\rm CC}(s)=K_{\rm CCP}+K_{\rm CCI}/s+\textup{j}\omega L_f$, \\
Power control (GFL) & ${G}_{\rm PC}(s)=K_{\rm PCP}+K_{\rm PCI}/s$,\\
PLL (GFL) & ${G}_{\rm pll}(s)=K_{\rm PLLP}+K_{\rm PLLI}/s$ \\
Voltage control (GFL, GFM) & ${G}_{\rm VC}(s)=K_{\rm VCP}+K_{\rm VCI}/s$\\
Frequency control (GFM) & ${M}(s)=\tfrac{1}{Js^2+Ds}$ \\
\midrule
Filter parameters & $R_f=0.05$, \; $C_f=0.06$, $L_f=0.05$\\
$Z_2=R_2+\textup{j}\omega L_2$ & $L_2=0.05$, \; $R_2=0.05$\\
$Z_3=R_3+\textup{j}\omega L_3$ & $L_3=0.1$, \; $R_3=0.1$\\
$Z_g=R_g+\textup{j}\omega L_g$ & $L_g=0.5$, \; $R_g=0.5$\\
\bottomrule
\end{tabular}
\end{table}

Integral and proportional constants are defined in Table \ref{tab:parameters}.
\begin{table}[ht]
\centering
\caption{Control converter parameters}
\label{tab:parameters}
\begin{tabular}{llll}
\toprule
\textbf{Type} & \textbf{Controller} & \textbf{Proportional} & \textbf{Integral} \\
\midrule
& Current control  & $0.06$   & $28.27$ \\
GFL & Power control    & $31.41$   & $246.74$ \\
 & Voltage control  & $31.41$   & $246.74$ \\
 & PLL              & $402.12$  & $4042.6$ \\
\midrule
GFM & Voltage control  & $0.031$   & $0.4$ \\
 & Current control  & $0.06$   & $28.27$ \\
 & Frequency control & $J=0.02$   & $D=0.1$ \\
\bottomrule
\end{tabular}
\end{table}

For the IEEE 14-node system, the converter parameters are given in Table~\ref{tab:parameters}. The PLL controller configuration is varied adjusting the PLL cutoff frequency. Specifically, we consider $f_{\text{pll}}=\{5,10,20,60\}$ for the converter in nodes 2,3,6, and 8, respectively. The selected $f_{\text{pll}}$ determines the proportional and integral gain according to $K_{\rm PLLP}=\omega_{\text{pll}}$ and $K_{\rm PLLI}=K_{\rm PLLP}^2/4$. Each GFL converter uses the corresponding PLL parameters listed in Table~\ref{tab:parameters_14node}.
\begin{table}[ht]
\centering
\caption{Control converter parameters for IEEE 14 node system}
\label{tab:parameters_14node}
\begin{tabular}{llll}
\toprule
\textbf{Converter} & \textbf{Proportional} & \textbf{Integral} \\
\midrule
 GFL 2  & $31.41$   & $246.74$ \\
GFL 3    & $62.83$   & $986.96$ \\
 GFL 6  & $125.66$   & $3947.7$ \\
 GFL 8  & $376.99$  & $35530.57$ \\
\bottomrule
\end{tabular}
\end{table}

\end{document}